\documentclass[10pt,twocolumn]{article}

\usepackage[nocompress]{cite}
\usepackage{url}
\usepackage{amsmath,amssymb,graphicx}
\usepackage{booktabs}
\usepackage{multirow}
\usepackage{times}

\newcommand{\mytitle}{Count-Free Single-Photon 3D Imaging with Race Logic}

\title{\mytitle \vspace{-0.1in}}

\author{
  Atul Ingle, David Maier \vspace{2pt} \\
  {\normalsize Portland State University}\\
  {\normalsize \texttt{\{ingle2, maier\}@pdx.edu}}
}
\date{}
\usepackage[per-mode=symbol]{siunitx}
\usepackage{nicefrac}
\usepackage{textcomp}
\usepackage{url}
\usepackage{array}
\usepackage{footmisc}
\usepackage{booktabs}
\usepackage{multirow}
\usepackage{amsthm}
\newtheorem*{theorem*}{Theorem}

\usepackage{geometry}
\usepackage[table,xcdraw,dvipsnames]{xcolor}

\usepackage[pagebackref,breaklinks=true,colorlinks,allcolors=cyan,bookmarks=false]{hyperref}

\newcommand{\rnote}[1]{\textcolor{red}{#1}}

\newenvironment{myitem}
{ \begin{itemize}
    \setlength{\itemsep}{0pt}
    \setlength{\parskip}{0pt}
    \setlength{\parsep}{0pt}     }
{ \end{itemize}                }

\usepackage{relsize}
\newcommand*{\defeq}{\stackrel{\mathsmaller{\mathsf{def}}}{=}}

\begin{document}

\maketitle

\let\rnote\relax

\begin{abstract}
Single-photon cameras (SPCs) have emerged as a promising new technology for high-resolution 3D imaging. 
A single-photon 3D camera determines the round-trip time of a laser pulse by precisely capturing the arrival of individual photons at each camera pixel.
Constructing photon-timestamp histograms is a fundamental operation for a single-photon 3D camera.  
However, in-pixel histogram processing is computationally expensive and requires large amount of memory per pixel.
Digitizing and transferring photon timestamps to an off-sensor histogramming module is bandwidth and power hungry.
Can we estimate distances without explicitly storing photon counts?
Yes---here we present an online approach for distance estimation suitable for resource-constrained settings with limited bandwidth, memory and compute.
The two key ingredients of our approach are (a) processing photon streams using \emph{race logic}, which maintains photon data in the time-delay domain, and (b) constructing count-free \emph{equi-depth histograms} as opposed to conventional equi-width histograms.
Equi-depth histograms are a more succinct representation for ``peaky'' distributions, such as those obtained by an SPC pixel from a laser pulse reflected by a surface.
Our approach uses a \emph{binner} element that converges on the median (or, more generally, to another $k$-quantile) of a distribution.
We cascade multiple binners to form an equi-depth histogrammer that produces multi-bin histograms.
Our evaluation shows that this method can provide at least an order of magnitude reduction in bandwidth and power consumption while maintaining similar distance reconstruction accuracy as conventional histogram-based processing methods.
\end{abstract}

\begin{figure}[!ht]
\footnotesize
\addtolength{\tabcolsep}{-3pt}
\begin{center}
\begin{tabular}{@{}rccc@{}}
\toprule
\multicolumn{1}{l}{} & \begin{tabular}[c]{@{}c@{}}Conventional\\ Analog LiDAR\end{tabular} & \begin{tabular}[c]{@{}c@{}}Histogram-based\\ SPCs\end{tabular} & \begin{tabular}[c]{@{}c@{}}Count-free\\ SPCs \textcolor{RedOrange}{[Ours]}\end{tabular} \\ \midrule
FPS (speed)          & \cellcolor[HTML]{32CB00}moderate-to-high                            & \cellcolor[HTML]{FFCE93}low-to-moderate                        & \cellcolor[HTML]{32CB00}high                                          \\
x/y-resolution       & \cellcolor[HTML]{FFCE93}low                                         & \cellcolor[HTML]{32CB00}high                                   & \cellcolor[HTML]{32CB00}high                                          \\
z-resolution         & \cellcolor[HTML]{FFCE93}low-to-moderate                             & \cellcolor[HTML]{32CB00}high                                   & \cellcolor[HTML]{32CB00}high                                          \\
z-range              & \cellcolor[HTML]{32CB00}high                                        & \cellcolor[HTML]{32CB00}high                                   & \cellcolor[HTML]{32CB00}high                                          \\
power                & \cellcolor[HTML]{FFCE93}high                                        & \cellcolor[HTML]{32CB00}low                                    & \cellcolor[HTML]{32CB00}low                                           \\
bandwidth            & \cellcolor[HTML]{32CB00}low                                         & \cellcolor[HTML]{FFCE93}high                                   & \cellcolor[HTML]{32CB00}low                                           \\ \bottomrule

\end{tabular}
\end{center}
\medskip

\centering
\includegraphics[width=0.95\linewidth]{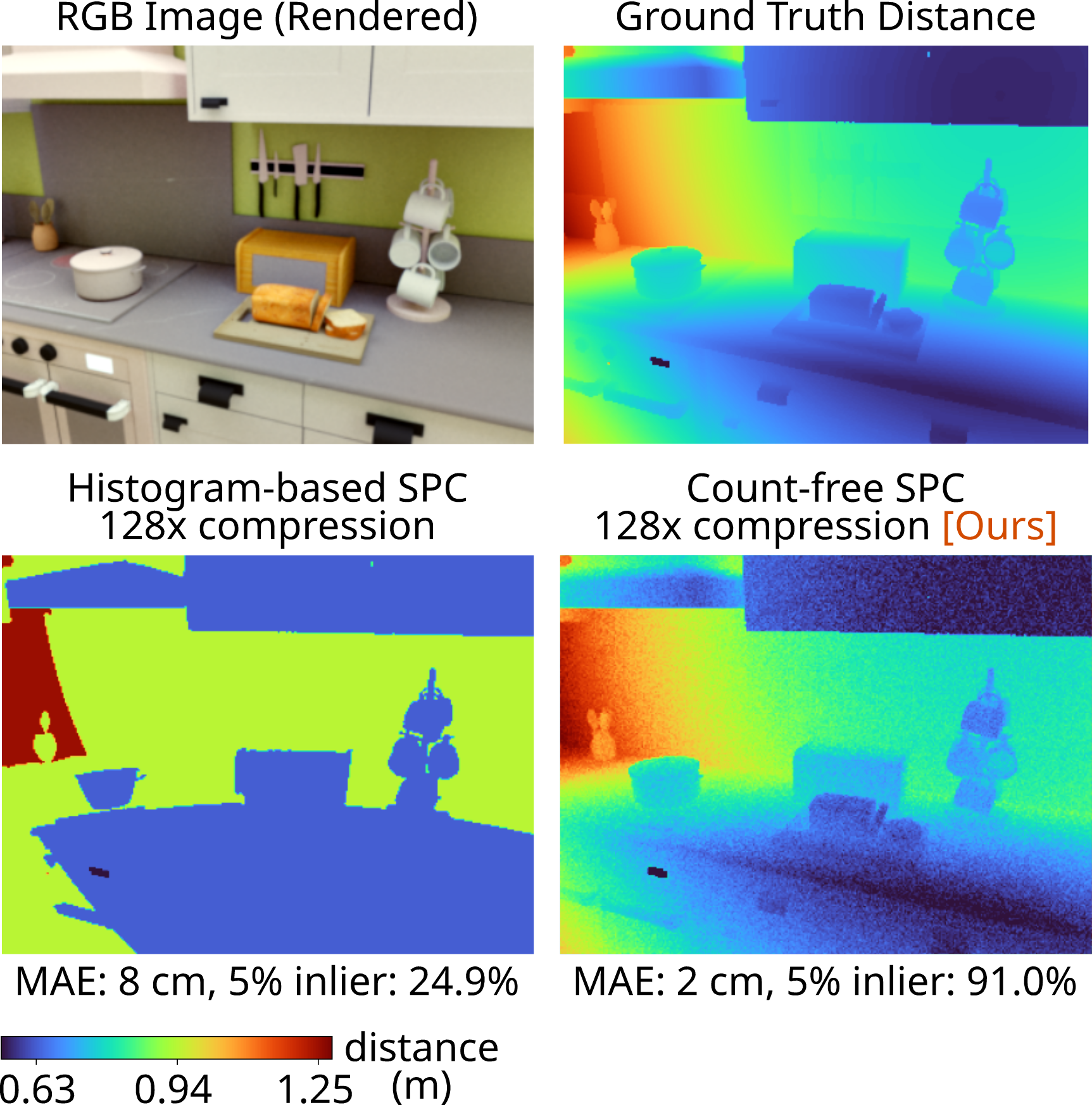}
\caption{\textbf{Comparison of direct time-of-flight 3D sensing techniques:} Single-photon 3D cameras enable higher spatial and distance resolutions compared to conventional analog LiDAR sensors that are based on low resolution arrays of photodiodes with spinning assemblies.
However, conventional histogram-based SPCs suffer from severe bandwidth bottlenecks that limit their frame rates.
Using coarsely binned histograms can reduce bandwidth requirement but at the cost of severe loss in distance resolution.
Our method is count-free and provides an on-sensor compressed representation suitable for distance estimation.
\label{tab:comparison}}
\vspace{-0.1in}
\end{figure}

\section{Introduction}\label{sec:introduction}

A wide range of computer vision applications---including industrial robotics, machine vision, autonomous driving, and augmented reality---need low-power 3D perception.
Image sensors capable of capturing single photons (e.g. single-photon avalanche diode (SPAD) sensors, high-gain avalanche photodiodes, and silicon photomultipliers) have gained popularity recently as detectors of choice for such applications.
SPAD-based 3D cameras have now found their way into commercial devices such as smartphone cameras \cite{Rangwala2020}, light detection and ranging (LiDAR) sensors for autonomous robotics \cite{Ouster2022} and cameras for scientific imaging \cite{Horiba}.
Due to their compatibility with CMOS fabrication technology, there is increasing availability of high (kilo-to-megapixel) resolution arrays of SPAD pixels with additional data processing embedded in the same hardware chip.
Unfortunately, their high sensitivity and speed is a two-edged sword: the amount of raw data generated by these sensors is orders of magnitude higher than can be reasonably processed or transferred in real time.
This aspect limits their applicability in many real-world applications, especially those that are power and bandwidth constrained.

\begin{figure*}[!ht]
  \centering
  \includegraphics[width=0.98\textwidth]{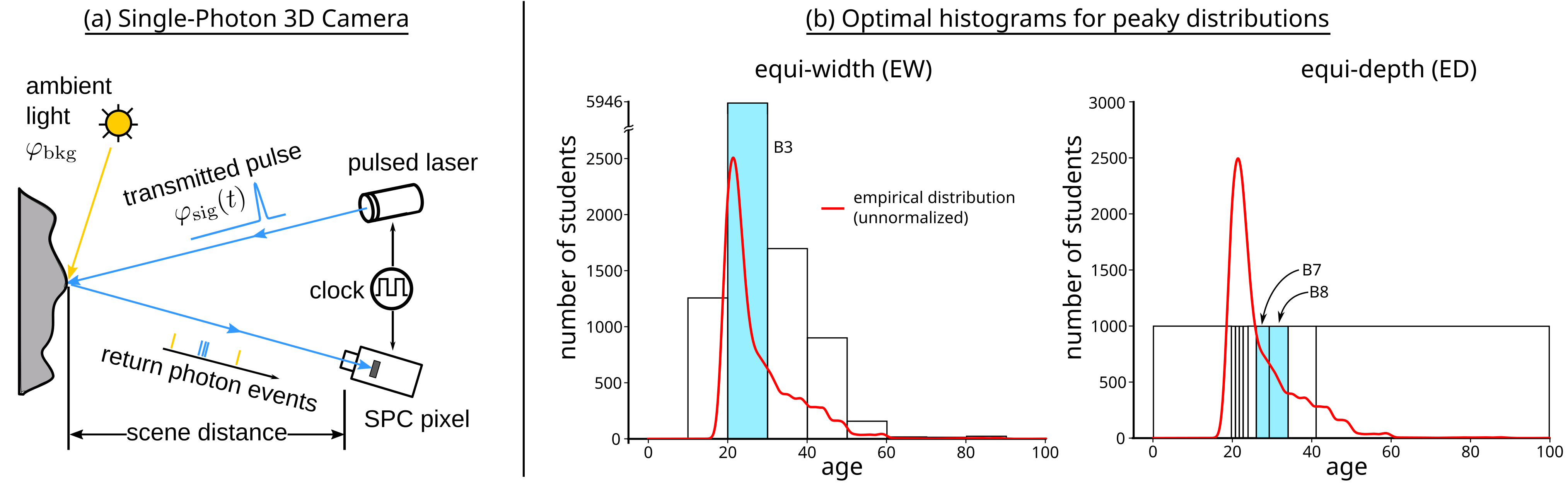}
  \caption{{\bf SPC image formation model and optimal histogramming for peaky distributions:}
  (a) An SPC captures a return stream of photon events at different time delays with respect to the transmission time of a laser pulse. 
  The photon stream contains both signal and ambient photons.
  (b) Data with peaky distributions are better summarized by an ED histogram because the narrow bins around the true peak location capture the shape of the peak more reliably. 
  Using bins B7 and B8 gives a more accurate estimate (543) of the number of students (554) in the age range 28 to 30 as opposed to using the single bin B3 (1189) in the EW histogram.
  \label{fig:intro_spc_edhist}}
  \vspace{-0.2in}
\end{figure*}

A single-photon camera (SPC) captures distance information using the direct time-of-flight (dToF) principle \cite{lange20003d}.
A pulsed laser source illuminates each scene point and the corresponding camera pixel captures a stream of \emph{return events} that represent the delays of the photon arrivals relative to the time the laser pulse was emitted.
These return events capture not only the returning laser photons but also spurious photons due to ambient light and other sources of noise.
Traditional methods estimate the time-varying distribution of light intensity, called the \emph{transient distribution}, by digitizing and storing photon return events in the form of an equi-width (EW) histogram where each histogram bin represents a fixed time interval, usually $\sim$10--100's of picoseconds.
The location of the peak in this EW histogram provides an estimate of the true distance of the scene point.
Although this EW histogram-based processing technique provides reliable distance estimates, time-to-digital conversion of individual photon delays and transferring these off-sensor for histogram formation consumes a large amount of power. 
It is infeasible to build high-resolution histograms on the image sensor due to the limited memory and computational resources available at each pixel. 
The inherent power- and bandwidth-hungry nature of histogram-based SPC data processing poses a significant hurdle to scaling this technology to higher pixel resolutions.


One way to compress EW histogram data is to use coarser histogram bins. However, as shown in Fig.~\ref{tab:comparison}, this approach causes severe quantization artifacts in the distance map.
We propose a radically different approach for direct time-of-flight imaging that enables high resolution 3D scene reconstruction while consuming orders of magnitude lower bandwidth. An example result achieving $>100\times$ compression is shown in Fig.~\ref{tab:comparison}.
Our technique relies on two key ingredients: (a) temporal encoding using race logic, and (b) equi-depth histograms.

\smallskip
\noindent\textbf{(a) Temporal encoding using race logic}:
Instead of attempting to digitize and store each individual photon delay using time-to-digital converters (TDCs) which often consumes a large fraction of the pixel area and total power, we perform as much of the processing as possible in the analog time-delay domain. 
To this end, we harness recent advances in the field of \emph{race logic}, where information is encoded not in the voltage levels of signals but in their precise arrival times \cite{Tzimpragos_2021}.
Temporal processing using race logic is naturally suited to single-photon time-of-flight 3D sensing because the arrival times of the photon-return events carry useful scene information (scene distance and reflectivity).
\rnote{Recent work \cite{zhang2022first,white2022differential} has shown the advantages of performing time-differential measurements instead of capturing
each photon arrival timestamps using time-to-digital converters (TDCs).
We avoid the need for TDCs by operating in the time-delay domain.}

\smallskip
\noindent\textbf{(b) Count-free equi-depth histograms}:
The second key ingredient of our method is to construct equi-depth (ED) histograms for summarizing the detected photon streams, rather than the conventional equi-width (EW) histograms that other methods employ (see Fig.~\ref{fig:intro_spc_edhist}(b)).\footnote{The word ``\textbf{depth}'' in equi-depth histograms should not be confused with scene point distances captured by a 3D camera.
In this paper, we will use the terms ``\textbf{distance}'' and ``\textbf{distance map}'' to avoid confusion.}
We present a technique that directly extracts ED histogram bin boundaries without explicitly storing photon counts, providing a parsimonious power- and bandwidth-efficient summary of the shape of the transient distribution.
Unlike an EW histogram, where most of the histogram bins are spent on storing ambient photons, our count-free ED histogram technique provides finer granularity around the peak of the transient distribution where signal photons arrive.

\rnote{
We combine these ingredients in a ``binner'' element (see Section \ref{sec:binners}) that tracks the median of a transient distribution.
Because of ambient light and asymmetry in the distribution, that median might be displaced from the peak of that distribution.
Therefore, we cascade multiple binners into a ``histogrammer'' that produces an ED histogram with multiple boundaries.
Our novel contribution is using race logic pre-processing to estimate equi-depth histograms in an online fashion while avoiding the large power and bandwidth requirements of current TDC-based designs.
}

\smallskip
\noindent \rnote{\textbf{Limitations:}
Although we do not show a complete hardware implementation of our proposed binner circuit, we believe the binner element is simple enough to fabricate on a sensor chip.
We note that a histogrammer will require a binner per ED histogram bin boundary which may increase circuit complexity.
However, if the required circuitry proves too extensive to include on a per-pixel basis, we can share it across pixels.
We show some possible pixel and pixel-array designs in Suppl.~Sec.~\ref{app:hardware}.
}



\section{Related Work}

\noindent\textbf{Single-photon 3D sensing:}
Almost all existing approaches for single-photon 3D imaging rely on explicitly constructing an EW histogram of photon arrival delays \cite{gyongy2021directTutorial} and extracting high-resolution distance information from them \cite{shin2016photon, rapp2021high, gupta2019photon, gupta2019asynchronous, patanwala2021high}.
Recent attempts at compressing this data rely on ideas from sketching \cite{sheehan2021sketching}, coding theory, and signal processing such as Fourier-domain compression \cite{gutierrez2022compressive}, random projections based on compressed-sensing theory \cite{colacco2012compressive}, and directly measuring a compressed representation such as a low-dimensional parametric model \cite{poisson20222}.
Hardware approaches involve on-chip resource-sharing schemes \cite{hutchings2019reconfigurable}, or two-step histogram capture (first acquire a coarse EW histogram, then zoom into a sub-region of interest \cite{lindner2018252, vornicu2019tof}).
The adaptive zooming strategy was also explored in a different context of distance resolution enhancement \cite{po2021adaptive}, where an SPC with a fast time-gate adaptively selects photons from a specific distance of interest.
Methods exist to compress EW histograms on the fly using a linear compression schemes that can be implemented in hardware.
However, on-sensor compression techniques often require additional in-pixel memory and compute to implement these compression matrices \cite{gutierrez2022compressive}.
Our work bypasses the need to explicitly form an EW histogram and directly generates a compressed representation of the full transient distribution, reducing the bandwidth and power requirements by at least an order of magnitude.
Our method does not store any coding matrices and can be implemented with minimal in-pixel memory and compute.


\smallskip
\noindent\textbf{Data-Stream and Time-Series Processing:}
Our work on compressing transient distributions is inspired by database literature on adaptive methods for characterizing data streams and time series.
Methods such as variable-window data-stream aggregation \cite{frames2016} and piecewise representations \cite{segmenting2001, shatkay1996approximate, himberg2001segmentation} can provide better task performance under constraints on the number of bins.
However, such methods are not suitable for on-line implementations because they make a complete pass over the entire history of data stream events.
In contrast, our method processes the photon delay stream on a per-laser-cycle basis, without explicitly storing the complete history of photon arrival events.
Scan statistics try to identify areas of above-average density in a dataset \cite{scanstats}, but these methods also suffer from limitations for large datasets \cite{preston-event-2009}.
Multiple authors have studied optimal histogram formulations over data streams, including approximate methods for histograms and related summaries with good time, space and error bounds \cite{guha-approx-2006, manku-approx-2002}.
However, applying these methods in our setting requires first converting delays to digital values, which we are trying to avoid.

\smallskip
\noindent\textbf{Race Logic:} Conventional digital logic encodes information in the form of specific voltage levels that denote binary digits (usually a low voltage for a 0 and a high voltage for a 1).
Instead of encoding information in the voltage levels, \emph{race logic} encodes information in the \emph{time of arrival} of a voltage-level change.
The term ``race'' originates from the idea of exploiting race conditions in digital circuits \cite{madhavan2015race}.
Signals encoded using their times of arrival can be treated as mathematical objects that lead to a  type of computational logic---different from conventional Boolean algebra---called \emph{space-time algebra} \cite{Smith2018}.
Since photon-arrival delays are the smallest unit of scene information captured by an SPC, the race logic paradigm is naturally suited to our problem.\footnote{The idea of using signal races has been used to implement Boolean functions \cite{lee2002race}.
In this work we use the term race logic in the sense of Madhavan et al. \cite{madhavan2014race} where the time delays themselves encode information.}
Recent work shows that standard digital circuits (e.g. \mbox{FPGAs}) can implement race logic and can provide orders-of-magnitude reduction in power while maintaining accuracy similar to digital counterparts \cite{Tzimpragos_2021}.


\section{Image Formation}\label{sec:binners}
Fig.~\ref{fig:intro_spc_edhist}(a) shows the image formation model for a single scene point using a dToF imaging system.
A pulsed laser sends out a short light pulse $s(t)$ with a repetition period $T$ into the scene.\footnote{A common choice for $s(t)$ is Gaussian pulse model with a known width (FWHM). This model can also include an instrument response function (IRF) $h(t)$ by convolving it with $s(t)$.}
Photons traveling at the speed of light $c$ bounce off the scene point located at a distance $d$ and are received by the corresponding pixel in the camera.
The same pixel also recieves a background signal $\varphi_\text{bkg}$, which we assume is a constant and does not vary with time.
The camera pixel receives a shifted and scaled version of this light pulse together with the background signal:
\begin{equation}
  \varphi(t) = \varphi_\text{sig}(t) + \varphi_\text{bkg}  \label{eq:trans_dist}
\end{equation}
where $\varphi(t)$ is periodic with period $T$, which is the spacing of the laser pulses.
Thus for $t > T$, $\varphi(t+nT) = \varphi(t)$ for $n \in \mathbf{Z}$.
For $0 \leq t < T$, $\varphi_\text{sig}(t) = \eta s(t-\tau_d)$.
The shift $\tau_d = 2d/c$ corresponds to the distance of the scene point and the unknown scaling $0<\eta<1$ accounts for scene reflectivity and signal loss terms due to imperfect optics and sensor quantum efficiency.\footnote{We assume that all scene points are located within a maximum distance range of $d_\text{max} = cT/2$.}
The ideal time-varying distribution of light intensity (Eq.~\ref{eq:trans_dist}) received by the camera pixel is called the transient distribution.
An SPC pixel captures a return stream of events consisting of photons that follow a periodic inhomogeneous Poisson process with rate $\varphi(t)$.
The total signal strength is given by $\Phi_\text{sig} \defeq \int_0^T \varphi_\text{sig}(t) \; dt$, the background strength is given by $\Phi_\text{bkg} \defeq \varphi_\text{bkg}T$, and the signal-to-background ratio is defined as $\textsf{SBR} \defeq \Phi_\text{sig} / \Phi_\text{bkg}$.
The background strength subsumes all sources of random noise that introduce unwanted counts (including ambient light and dark counts).

\smallskip
\noindent {\bf Conventional SPC (equi-width histograms):}
A conventional SPC pixel uses a time-to-digital converter (TDC) to digitize and aggregate photon arrival times into an EW histogram of photon counts.
Assuming a time bin resolution of $\Delta$, the histogram has $B \defeq \lceil T/\Delta \rceil$ bins.
Most practical systems today construct such histograms with $\sim$1000 time bins, where each bin stores photon counts as 8-bit unsigned integers.
EW histograms are constructed for each scene point by repeating the measurement process over many ($\sim$100--1000's) laser pulses.
Intuitively, the location of the peak of this measured histogram provides an estimate for the true distance.
In practice, distance is estimated by picking the ``arg max'' bin of this histogram or using techniques such as log-matched filtering in combination with statistical \cite{shin2016photon} or data-driven \cite{lindell2018single} spatio-temporal processing methods.

Although this approach provides high-quality distance maps, the process is resource hungry: Imagine a megapixel SPC with dedicated 1024-bin histogrammers in each pixel, where each histogram bin maintains a one-byte photon-count register.
Such a camera running at 30 fps would generate over 200 gigabits per second, an impractically large amount of data to transfer off the image sensor.
An alternative is to stream each digitized timestamp off the image sensor to a dedicated histogrammer circuit.
This alternative is undesirable because of the need to move large amount of timestamp data, which not only consumes power but also places a large bandwidth bottleneck at the image sensor.

What if, instead, we could capture distance information without explicitly capturing photon-count histograms?


\smallskip
\noindent {\bf Towards count-free histograms:}
The basic building block of our method is a \emph{binner circuit} shown in Fig.~\ref{fig:binner_ckt}(a) that probabilistically tracks a certain $k$-quantile (usually the median) of the transient distribution.
As shown in Fig.~\ref{fig:binner_ckt}(b), the binner consists of a reference signal (RS) that divides the incoming photon stream (SR) into early and late streams (SE and SL), and a control value (CV) that determines the duration of RS and which the binner adjusts during operation.
The reference signal drives a PASS-INHIBIT operation found in race logic \cite{Tzimpragos_2021,Smith2018}.
In each laser cycle, the binner circuit tracks the number of photons in SE and SR and and adjusts the control value to move the reference-signal boundary in the direction where more photons were received, with the objective of eventually having equal numbers of events on either side of this boundary.

The movement of the bin boundary is probabilistic and can be modeled as a random walk.
Suppose that the binner operates over a window length of $L$.
\rnote{We assume an underlying discrete time-grid; the binner's CV is updated by discrete step sizes of $\pm 1$.
In practice, the SPAD pixel's jitter limits the smallest possible step size to $\sim 100$~picoseconds \cite{sun2019simple,shimada2021back}.}
Let $E_i$ denote the CV after $i$ laser cycles have elapsed.
The number of photons in SE and SL are independent Poisson random variables (with possibly unequal means depending on the current CV location).
Hence, the transition probabilities $\mathsf{P}(E_i=k\pm 1 | E_i=k)$ are given by a Skellam distribution. (See Supplement \ref{app:markov} for details).
After sufficient laser cycles have elapsed, the limiting distribution of $E_i$ is the stationary distribution of this Markov chain.
The shape and spread of this limiting distribution can be computed numerically; it is a function of the true peak location, the signal strength and SBR.
Although we are not guaranteed movement towards the true median on every cycle, the mode of the stationary distribution (for realistic SBR levels) is at the overall median of the transient distribution (Fig.~\ref{fig:statn_dists}).\footnote{We use the terms ``median-tracking'' and ``quantile-tracking'' in this probabilistic sense.
The binner circuit does not ``lock on'' to the exact median of the transient distribution because it does not store the full history of photon arrivals and those arrivals are random in any given cycle.}

The following theoretical result shows that, with high probability, CV moves closer to the median on subsequent laser cycles if its current position is sufficiently far from the true median.
\begin{theorem*}
Let $\epsilon >0$ be an arbitrarily small probability threshold and suppose that CV $\neq$ true median. Then, for sufficiently large total photon rate (signal+background), the probability that CV does not move closer to the median on the next cycle is $< \epsilon$.
\end{theorem*}
\noindent The proof uses a Chernoff bound argument for the difference between two independent Poisson random variables. See Supplement \ref{app:markov} for details.

A key feature of the binner approach is that it is \emph{count-free}.
Therefore our theoretical guarantee is weaker than previously known results for streaming median estimation \cite{guha2009stream, munro1980selection}, which require maintaining some history of past photon timestamps.
We do not store the history of photon counts across cycles to form a histogram.
The binner updates its CV immediately and locally in each laser cycle based only on photons received in that cycle.
Each pixel need only maintain its CV, providing a large reduction in data requirements.
While our simulations represent the CV as a register, our method is also compatible with purely analog implementations without in-pixel time-to-digital conversion, using, for example, a comparator that outputs an analog quantity proportional to difference in the number of photons in SE and SL. 
Either approach allows us to adjust the CV on a finer scale than what is afforded by the fixed-size bins of an EW histogram.

A median-tracking binner generates a two-bin equi-depth histogram where the CV corresponds to the estimated location of boundary between the two bins.
In an ideal scenario of extremely high SBR (negligible background strength), the median will closely track the true peak, and we can estimate distance by simply reading out and digitizing CV.
More generally, a quantile-tracking binner circuit can track other quantiles of the transient distribution, by changing its increment-decrement logic.
For example, making CV increments $3\times$ as large as decrements, the binner can track the $75^\text{th}$ percentile of the transient distribution.

\smallskip
\noindent\textbf{Effect of Photon Pileup:}
\rnote{A SPAD pixel cannot capture incident photons that arrive too close to each other because it needs time to reset after each photon detection.
This effect is called the \emph{dead-time} and is on the order of $10$-\SI{100}{\nano\second}.
In case of extremely high incident photon flux, the measurements are biased towards early arriving photons, which results in pileup distortions.
Existing computational pileup correction techniques that operate on EW photon histograms \cite{coates1968correction,pediredla2018signal} cannot be used for the output of a binner because it does not keep track of the full history of photon counts.
However, our method is compatible with existing optical and hardware approaches  \cite{beer2018background,gupta2019asynchronous} that mitigate both ambient- and signal-induced pileup during acquisition.
}

\begin{figure}[!ht]
  \centering
  \includegraphics[width=0.80\linewidth]{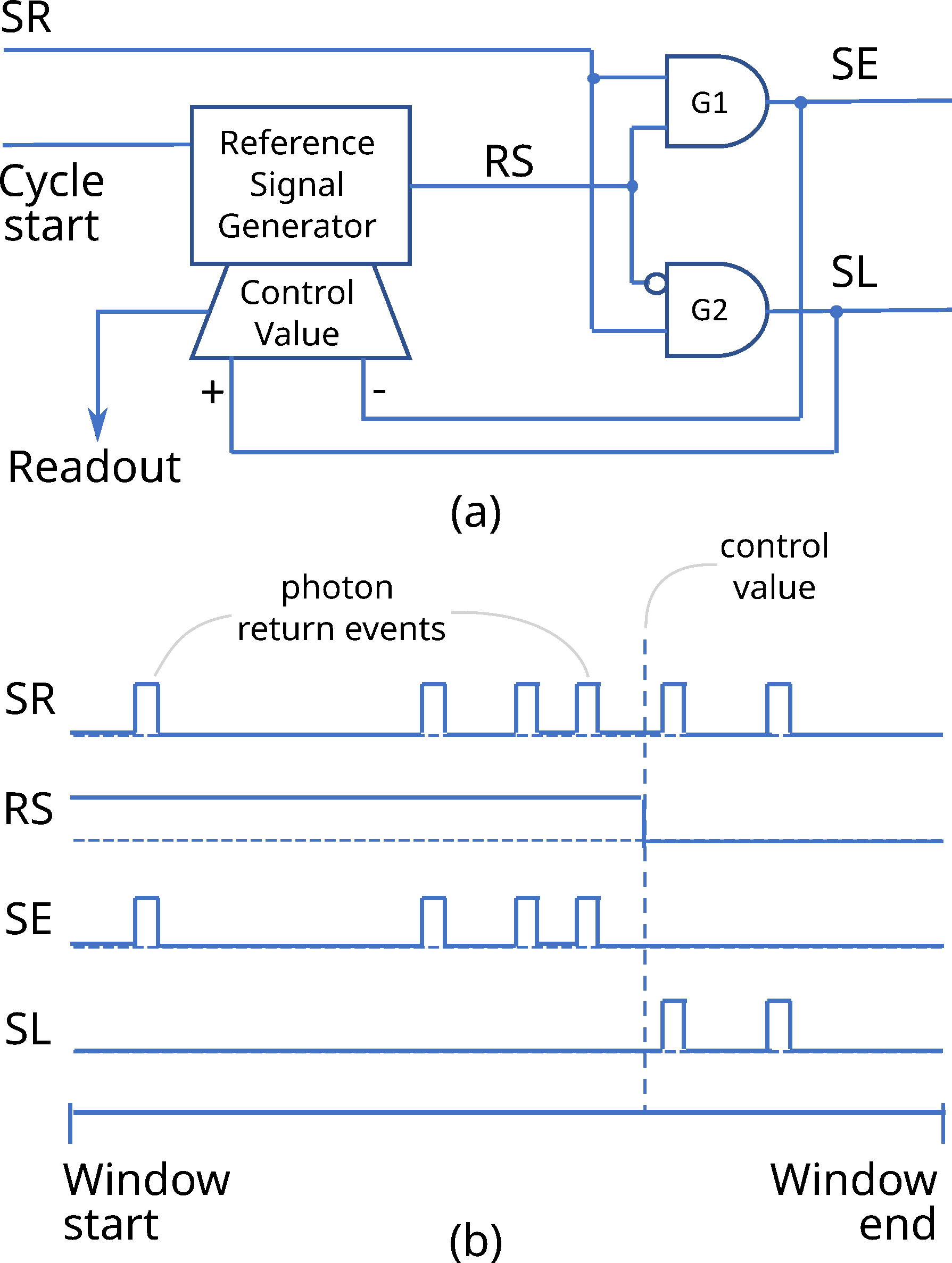}
  \caption{ {\bf
      Proposed ``binner'' circuit tracks the median (or other quantiles) of a transient distribution.}
    (a) The basic binner circuit receives a photon event stream as its input which it splits into early and late events.
  The difference between the two sides is used to update a control value that converges towards the true median over multiple laser cycles.
(b) An example timeline showing the photon stream that is split into early and late streams using a race logic-inspired reference signal that applies an INHIBIT-PASS operation around the current control value. \label{fig:binner_ckt}}
\end{figure}

\section{Analysis of a Median-Tracking Binner}
In this section we explore binner behavior---convergence, bias and accuracy---via simulation and formal models.
Our simulation produces cycles of photon return events sampled from an underlying transient distribution.
(In this section, the transient distribution is based on a Gaussian pulse plus random ambient light; in later sections we also use transient distributions generated from 3-D scene models.)
For each cycle, we sample this distribution using an inhomogeneous Poisson process that simulates shot noise, dark counts and quantum efficiency of a real-world SPAD pixel.
For our simulation, the range, $B$, is 1000 units and we adjust the CV in steps of size 1.
We also use a Markov-chain model, detailed in Supplement \ref{app:markov}.

\begin{figure}[!ht]
    \centering
    \includegraphics[width=0.95\linewidth]{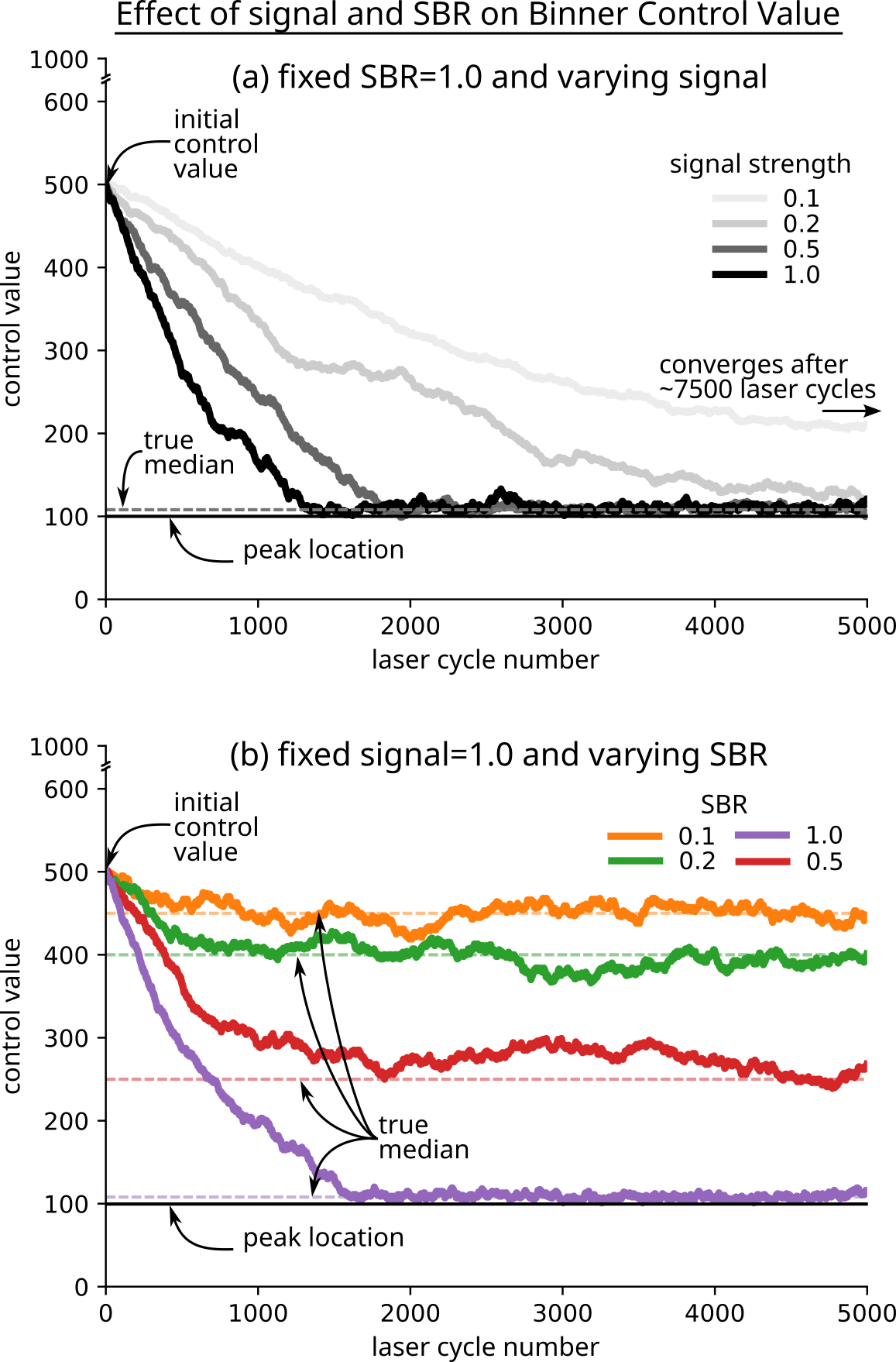}
    \caption{{\bf Simulated trajectories of a binner's control value for varying signal strength and SBR.}
    We plot the evolution of a binner's control value as a function of laser cycle number under different signal and SBR conditions. The ground truth transient consists of a laser pulse with FWHM $\SI{2}{\nano\second}$, located at $100$ with $B=1000$. The binner step size is set to $\pm 1$. (a) 
    For fixed SBR, convergence is faster for higher signal strengths. (b) For fixed signal strength, as SBR increases, the true median location moves closer to the true peak location.}
    \label{fig:example_binner_runs}
\end{figure}

\noindent \textbf{Binner convergence:}
The CV is an estimate of the true median of the transient distribution (including background).
If CV is away from true median, then more photons will likely arrive on the side towards true median, so it is more likely that CV will move towards the true median than away.
Intuitively, in the absence of noise, the binner's CV must settle at the location where the numbers of events in the SE and SL streams (Fig. 4) are---on average---equal, which is the median.
Once CV reaches true median, it will fluctuate around it.

The speed of binner convergence depends on the signal strength and SBR.
Fig.~\ref{fig:example_binner_runs}(a) shows the CV for example runs of a simulated binner over multiple laser cycles with fixed SBR=1.0 and varying signal strength from low to high.
We see in general that convergence is faster with higher signal strength (since a cycle is less likely to have balanced or no photons), but is slowed by higher background levels of (since there is a greater chance of a move in the wrong direction).
These runs used a fixed step size of $\pm 1$.
See Supplement \ref{app:convergence} for initial results with stepping strategies that can help speed up convergence.

\smallskip
\noindent \textbf{Bias due to ambient light:}
In low ambient light, a binner's CV moves towards the location of the peak of the transient distribution and eventually settles  at or near it.
However, with strong ambient light, the final boundary can be farther away from the peak because the median of the transient distribution depends on not only the signal but also the background.
The stronger the background, and the farther a peak is from the midpoint of the range, the larger the bias.
Fig.~\ref{fig:example_binner_runs}(b) shows example runs of a binner under fixed signal and varying SBR.
The shift in the true median is apparent.
At high SBR, the median is close to the peak position, and ``wandering'' is more constrained.
As the background increases (SBR decreases), the median shifts from the peak position towards the midpoint ($500$) of the total range.
While it might be possible to adjust for bias if signal strength and background are known, in Sec.~\ref{sec:edh} we present a more robust approach to this issue. 

\smallskip
\noindent \textbf{Binner accuracy:}
We saw that once a binner reaches the vicinity of the true median, its CV can continue to ``wander'' around that position.
Thus, even after convergence, there is still a range of possible values we might obtain on read-out of the CV.
Several factors influence how much deviation we expect to see from the true median, such as signal strength, SBR and the position of the signal peak.
Supplement \ref{app:markov} provides a Markov-chain model of binner behavior, which we have used to study convergence and accuracy.
While we do not have a closed-form solution for the stationary probability distribution (SPD) of this model, we can derive SPDs numerically for different combinations of signal strength, background level and peak position.
These SPDs provide insight on the degree of ``wandering'' of the CV under different conditions.
Fig.~\ref{fig:statn_dists} shows a variety of SPDs for different signal-peak positions (100, 250, 400) and SBRs (1.0, 0.5, 0.2, 0.01).
For each SPD, the gray line indicates the true median of the corresponding transient distribution.
Some observations:
\begin{myitem}
    \item The modes of the SPDs are always at the true medians.
    \item The true median is closer to the range midpoint with decreasing SBR.
    \item The spread of the SPD around the true median is greatest for intermediate SBRs.
    At those ratios, the transition probabilities for the CV roughly balance to the left and right.
\end{myitem}
We quantify the spread for these SPDs in Table \ref{tab:statn_dist_prob}, which shows the probability of the CV being within 5, 10 and 20 units of the true median (after convergence).
One means to cope with the ``spread'' of CV values is to read out the CV multiple times.
We evaluate that strategy in Section \ref{sec:edh}.

\begin{figure}[!t]
  \centering
  \includegraphics[width=0.99\linewidth]{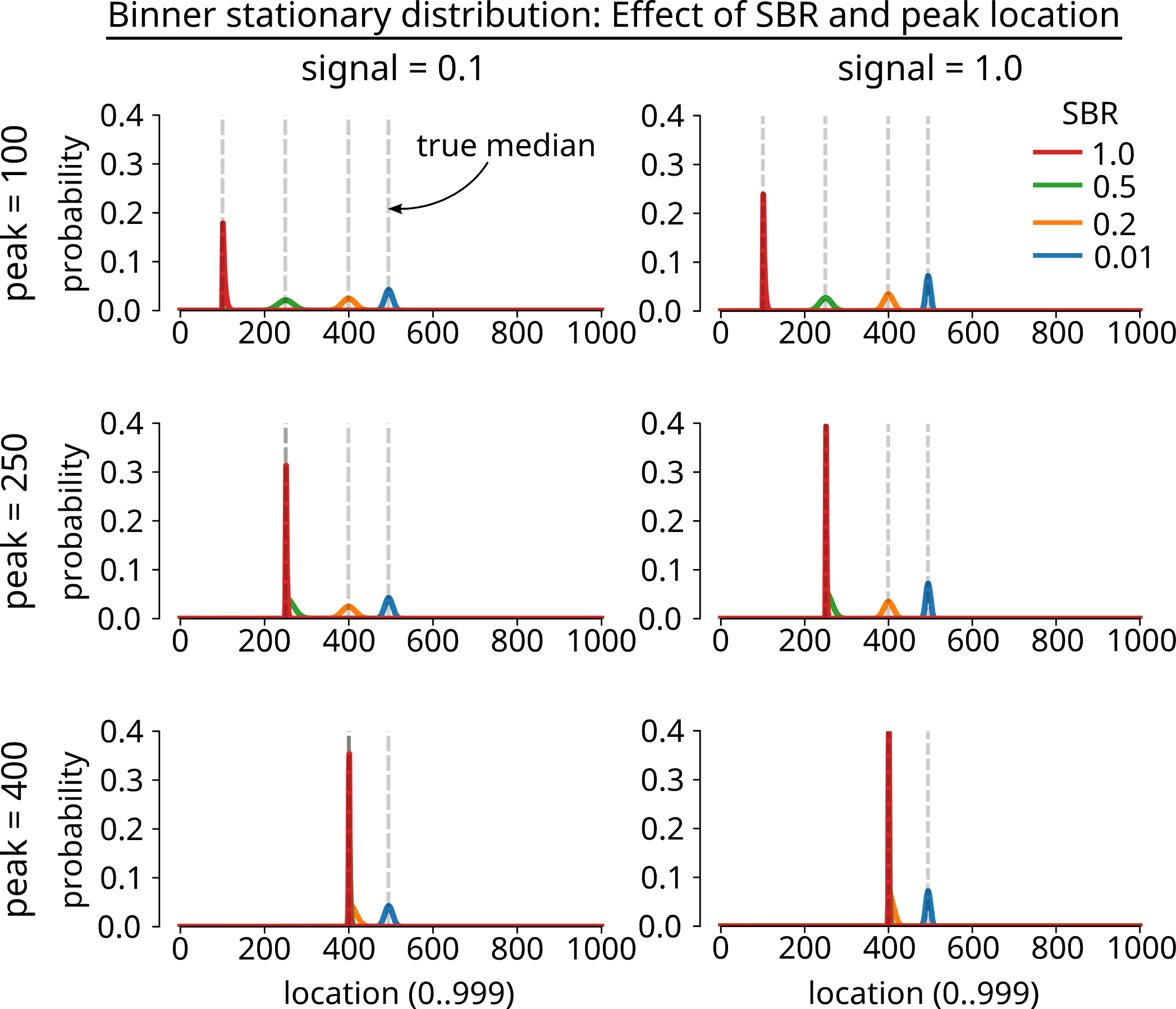}
  \caption{ {\bf
    Stationary distributions of a binner's control value at different signal and background levels and varying peak locations as predicted by the Markov chain model.}
    The median of the transient distribution is closer to the true peak location at high SBRs and closer to the midpoint of the overall range as the background increases.
  Observe that the mode of the stationary distribution is always aligned with the true median location which suggests that the control value probabilistically tracks the true median. The spread depends on the peak location and SBR.
\label{fig:statn_dists}}
\end{figure}

\begin{table}[!ht]
\footnotesize
\addtolength{\tabcolsep}{-1.5pt}
\begin{center}
\begin{tabular}{@{}cccccc@{}}
\toprule
                                                        & \multicolumn{1}{l}{}                                   & \multicolumn{4}{c}{\% probability $\pm 5/\pm 10/\pm 20$ from median}                                                                                                                                                              \\ \cmidrule(l){3-6} 
\begin{tabular}[c]{@{}c@{}}peak\\ location\end{tabular} & \begin{tabular}[c]{@{}c@{}}signal\\ level\end{tabular} & \begin{tabular}[c]{@{}c@{}}SBR \\ 0.01\end{tabular} & \begin{tabular}[c]{@{}c@{}}SBR \\ 0.2\end{tabular} & \begin{tabular}[c]{@{}c@{}}SBR \\ 0.5\end{tabular} & \begin{tabular}[c]{@{}c@{}}SBR \\ 1.0\end{tabular} \\ \midrule
\multirow{2}{*}{100}                                    & 0.1                                                    & 40/71/97                                            & 24/46/78                                           & 21/41/72                                           & 64/89/99                                           \\
                                                        & 1.0                                                    & 63/93/100                                           & 34/62/92                                           & 27/50/83                                           & 76/95/100                                          \\ \midrule
\multirow{2}{*}{250}                                    & 0.1                                                    & 40/71/97                                            & 24/46/78                                           & 20/40/71                                           & 92/99/100                                          \\
                                                        & 1.0                                                    & 63/93/100                                           & 34/62/92                                           & 24/49/82                                           & 97/100/100                                         \\ \midrule
\multirow{2}{*}{400}                                    & 0.1                                                    & 40/71/97                                            & 25/47/79                                           & 89/99/100                                          & 98/100/100                                         \\
                                                        & 1.0                                                    & 63/93/100                                           & 32/61/92                                           & 96/100/100                                         & 100/100/100                                        \\ \bottomrule
\end{tabular}
\end{center}

\caption{Probability of the binner control value is within $\pm 5/\pm 10/\pm 20$ units of the true median location after Markov chain convergence at two different signal strength of $0.1$ and $1.0$ and varying SBRs.
\label{tab:statn_dist_prob}
}

\end{table}

A single binner might not provide a good estimate of the peak of the transient distribution, because of bias and greater ``wandering'' resulting from ambient light.
Asymmetric and multiple peaks also present challenges.
The next section describes how a cascade of binner circuits can capture multi-bin equi-depth histograms of the transient distribution.
We will see that the adaptive nature of ED histograms can cope with bias and increase SBR in the critical region of the signal.
We show that it is possible to obtain accurate distance estimates with as few as 8 or 16 ED histogram bins due to their adaptive nature and ability to reliably focus attention around the true peak, as shown intuitively in Fig.~\ref{fig:intro_spc_edhist}(b).

\begin{figure*}[!ht]
  \centering
  \includegraphics[width=1.00\textwidth]{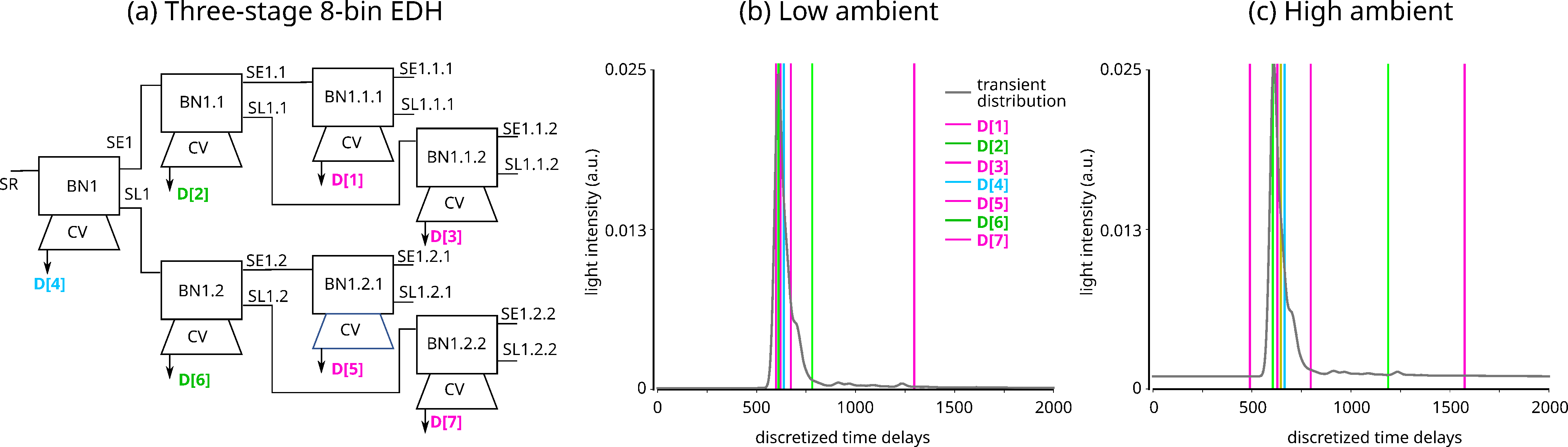}
  \caption{ {\bf Equi-depth histogrammer consisting of a tree of binners.}
    (a) Binners shown in Fig.~\ref{fig:binner_ckt} can be cascaded into a 3-stage binary tree.
    An inorder traversal of this tree gives the 7 ED bin boundaries $D[1..7]$.
    The binner outputs of an 8-bin (7-boundary) EDH for a simulated transient distribution are shown in (b) and (c).
    (b) In the presence of low ambient light, most of the ED bin boundaries cluster at the true peak location in this simulated distribution.
    (c) In case of higher ambient light, the bin boundaries are spaced farther apart with additional bins absorbing the ambient light, but the narrowest ED bins still cluster around the true peak location.
  \label{fig:EDH_and_ambient}}
  \vspace{-0.1in}
\end{figure*}

\section{Count-Free Equi-depth Histogrammer\label{sec:edh}}
This section describes a method for cascading multiple binner circuits to capture a multi-bin ED histogram.
We call this implementation an \emph{equi-depth histogrammer} (EDH).
In this binary-tree arrangement of binners, the early and late streams from each binner are fed into two binners in the next stage of the tree.
For illustration, Fig.~\ref{fig:EDH_and_ambient} shows a 3-stage EDH (that produces eight bins) and some example outputs.
The seven bin boundaries $D[1..7]$ correspond to an inorder traversal of the EDH tree outputs.

\smallskip
\noindent {\bf Robustness to ambient light:}
An advantage of the tree-based approach is that binner stages at deeper levels in the tree that are at or near the true peak location are not significantly affected by the presence of ambient photons outside the sub-windows that are enforced by the higher levels in the tree.
This property not only provides ambient-light rejection, but also optimizes readout bandwidth.
Under realistic SBR levels, upper binner stages of an EDH partition off background light, while most of the binners automatically and adaptively focus on the interesting regions of the transient distribution where the peak is located (Fig.~\ref{fig:EDH_and_ambient}(c)).
In contrast, conventional EW histograms allocate equal number of bins throughout the measurement window.
When using an EW histogram in the presence of strong ambient light, a significant fraction of the total readout bandwidth and energy is spent on digitizing and transferring background photons. 
Although this aspect could be improved by focusing EW histogramming resources around the true peak location, \emph{a priori} knowledge of the peak location may not be available.
An EDH focuses its bins around the true signal peak without prior knowledge of its location.

The adaptive nature of EDHs also helps limit CV ``spread'' at low SBRs.
A narrow bin around a signal peak will have a higher SBR than the overall ratio. (``Background'' bins will have lowered SBR.)
As a notional example, consider an EDH with a window width of $B = 1000$ where both signal and background levels are at 2.0, so SBR is 1.
Consider a stage-4 binner whose input is a narrow bin from its stage-3 parent.
Suppose that bin is 50 units wide.
At stage 3, there are 8 bins, and the total level is 4.0, so each bin is expected to have 4.0/8 = 0.5 units.
Since the background level is uniform across the whole range, the background contributes $(50/1000) \cdot 2.0 = 0.1$ units to the bin.
Thus, the other $0.5 - 0.1 = 0.4$ units must be signal, hence the SBR for the bin is 4.
As another example, if the stage-3 bin is 20 units wide, we will have 0.04 units of background in the bin and 0.46 units of signal, for an SBR of 0.46/0.04 = 11.5.
Hence we can expect less variance in boundary positions for the narrow bins around the signal peaks as we go down the EDH tree.


At this point it is natural to wonder:
How should the different stages of an EDH be initialized?
How many laser cycles does it take for the different stages to converge? 
How do we estimate scene distance from the ED histogram bin boundaries obtained from an EDH?
And finally, how do different signal and background strengths affect performance?
We now analyze practical design aspects of an EDH, then turn to performance evaluation in Sec.~\ref{results}.

\smallskip
\noindent{\bf EDH initialization and updates:} 
Assuming no \emph{a priori} information about the true peak location or the shape of the transient distribution, at the outset, we initialize the CV of the first stage BN1 at half of the maximum distance range. 
After a certain number of laser cycles have elapsed, we freeze BN1 and initialize the next stage of binners BN1.1 and BN1.2 at the midpoints of their respective sub-ranges and then adjust their CVs over subsequent laser cycles.
For example, if binner BN1 has a CV $D[4]$ when frozen, we initialize BN1.1 and BN1.2 to $D[4]/2$ and $(T+D[4])/2$, respectively, where $T$ is the overall window size (laser cycle duration).
In general, we launch binners at stage $i+1$ after freezing the binners at stage $i$, initializing each of the former at the midpoint of the range prescribed by the binner feeding it from the previous stage.
\rnote{In our simulations, we typically run a 4-stage EDH for 5000 cycles; each stage runs for 1250 cycles.}

The CV of a given binner are adjustable on a finer time scale than for TDC-based systems, since we do not digitize photon arrivals.
When seen from the perspective of the bin locations of a high-resolution EW histogram, the EDH control values can span ``fractional'' bins and lie between two bin edges.
Although a fundamental limit is imposed by the timing jitter of the SPC pixel ($\sim$$\SI{30}{\pico\second}$ for SPADs), we can allow the control value to adjust on an even smaller timescale to average out the effect of this jitter.

\smallskip
\noindent{\bf EDH convergence:}
The Markov chain analysis of convergence of a single stage binner (Supplement \ref{app:markov}) can be extended to each stage of the multi-binner EDH.
Convergence is probabilistic.
While we expect a binner's CV to generally move towards the true median of its sub-range, on any given laser cycle the photon events might split exactly around CV, or may even be higher on the side away from the median.
So there will not necessarily be progress on every cycle.
Stronger peaks lead to faster convergence.
A stronger signal peak means more return events expected in a cycle, meaning they are more likely to to reflect the distribution, and hence result in a step in the correct direction for CV.
Higher background levels will slow convergence.
It is possible to accelerate convergence towards the median by using a larger adjustment to the binner CV in the initial laser cycles and gradually reducing the step size towards the end.

\smallskip
\noindent{\bf Distance estimation:}
The width of each bin in an ED histogram is, by definition, inversely proportional to the local density of the underlying transient distribution.
Assuming that the transient distribution has exactly one sharp peak, distance of the scene point can be estimated by simply locating the narrowest ED bin (highlighted in red in Fig.~\ref{fig:argmax_curvefit}).

\begin{figure}
    \centering
    \includegraphics[width=0.7\linewidth]{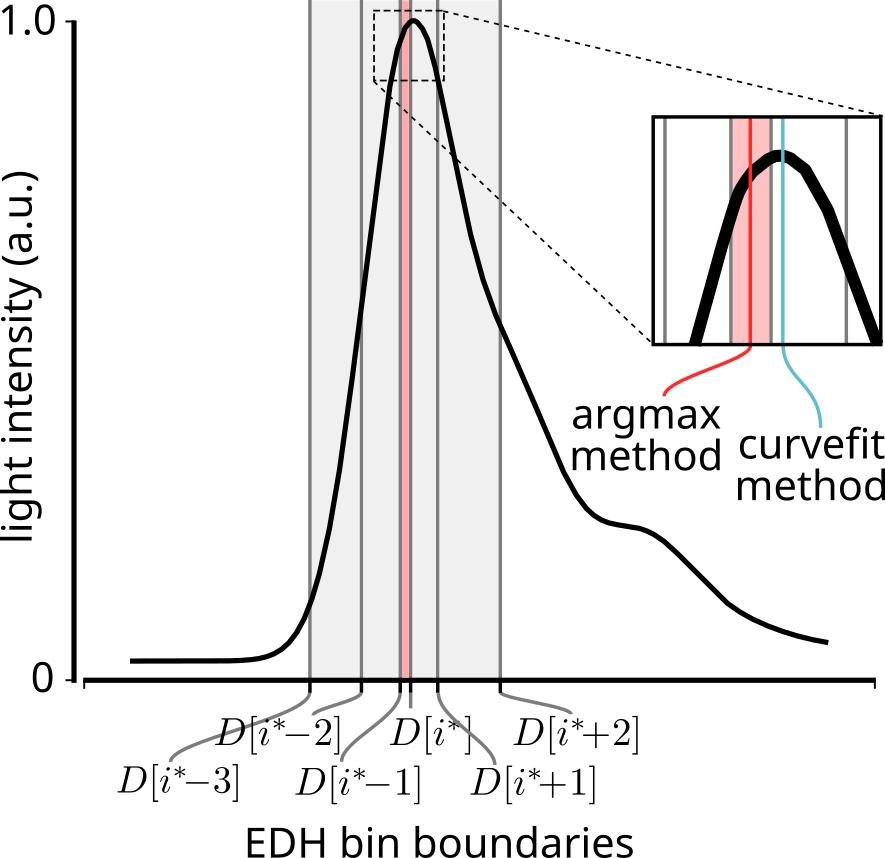}
    \caption{{\bf Two methods for distance estimation:}
    The ``argmax'' method simply locates the narrowest bin and picks the midpoint as the distance estimate.
    Note that this choice may not always locate the true peak (as in this case) when a boundary splits the peak.
    A curve-fitting approach uses the widths of multiple bins in the local neighborhood of the narrowest bin.
    \rnote{Note that these computations occur off-sensor, after the ED bin boundaries have been read out.}
  }
    \label{fig:argmax_curvefit}
\end{figure}

Suppose the tree-based EDH consists $K$ stages giving $2^K-1$ bin boundaries $D[1,2,3, \ldots, 2^K-1]$.
The first and last bin edges are, by definition, located at the extreme ends of the time window, i.e., $D[0]=0$ and $D[2^K]=T$.
Let $i^* = \arg \max_{1 \leq i \leq 2^K} \frac{1}{D[i]-D[i-1]}$ be the right edge of the narrowest bin.
The ``argmax'' distance estimator is the midpoint of the narrowest ED bin:
\begin{equation}
\widehat{d}_\text{argmax} \defeq \frac{c}{4} \left(D[i^*]-D[i^*-1]\right) \label{eq:argmax_estimator}
\end{equation}
where $c$ is the speed of light. 
We can extend this method of distance estimation to handle multiple peaks (say due to multiple reflections, or presence of semi-transparent materials along the viewing direction) by replacing the argmax operation by a more general peak finding routine which may return locations of all ``locally narrow'' bins (those narrower than adjacent bins).

Although a locally narrow bin provides a first order estimate of the location of the true signal peak, there are situations where two ED bins split the peak in a way that the midpoint of the narrowest bin gives a biased estimate of the peak location.
This effect may be exacerbated in real-world settings where the laser peak is often not symmetric and has a sharper leading edge and a longer trailing edge.
To handle such cases, we also use a quadratic curve fitting method that provides finer estimates of the peak location.
Denoting $x_i = (D[i]-D[i-1])/2$ and $y_i = \frac{1}{D[i]-D[i-1]}$, we fit a quadratic $y=\alpha x^2 + \beta x + \gamma$ using the $(x_i, y_i)$ pairs for ED histogram bin indexes surrounding the narrowest bin.
The number of bins on either side of the narrowest bin is chosen adaptively to lie within one-standard deviation of all the ED bin widths.
We have found that this usually results in a subset of the bins $\{i^*-2, i^*-1,i^*, i^*+1, i^*+2\}$ being chosen for curve fitting (shaded gray in Fig.~\ref{fig:argmax_curvefit}).
The scene distance is estimated using: 
\begin{equation}
  \widehat{d}_\text{curvefit} = -\frac{c\beta}{4\alpha}.
\end{equation}
\rnote{Other more sophisticated distance estimation routines that use ED histogram boundaries as inputs are also possible.
Due to in-pixel memory and compute constraints, we envision these will be implemented in an off-sensor compute module.}

The next section provides simulation results for distance estimation in various signal and background levels, and distance-map reconstructions of rendered scenes. 

\section{Results}\label{results}
We evaluate our EDH-based method and compare it against conventional EW histogramming methods using single-pixel simulations and a transient rendering dataset.
The single-pixel simulations cover a range of signal/background strengths and ground truth distances.
The rendered scenes provide greater variety in the shapes of the transient distributions due to multi-path effects.
Finally, we test the method with real-world hardware data.

\begin{figure}
\centering
\includegraphics[width=0.95\linewidth]{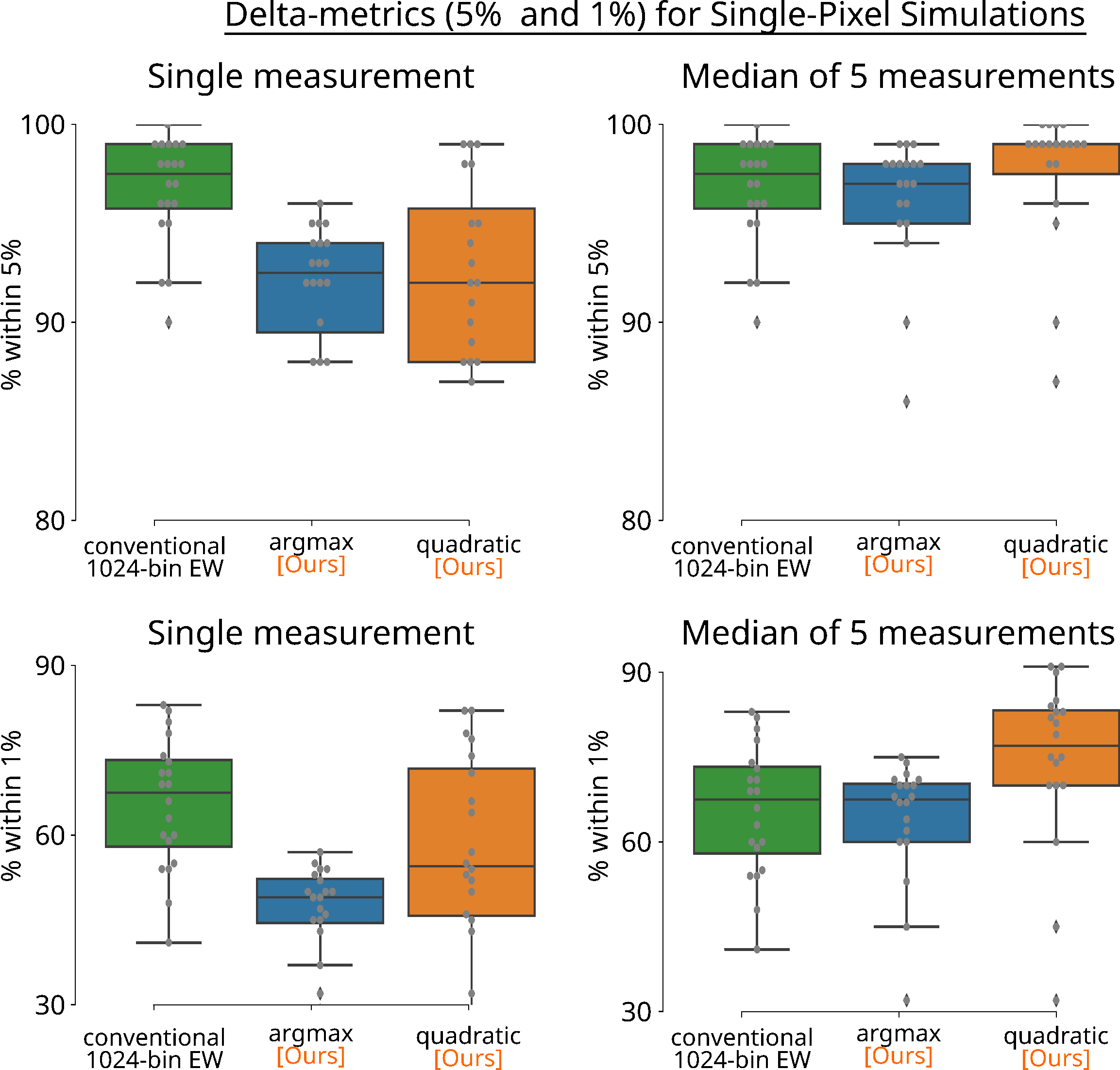}
\caption{ {\bf Single-pixel Monte Carlo simulation results.} These boxplots aggregate information over the range of signal and background strengths for a laser pulse width of $\SI{5}{\nano\second}$ with 100 randomly chosen ground-truth distance values.
\label{fig:singlepixel_montecarlo}}
\end{figure}

\begin{figure*}[!ht]
  \centering
  \includegraphics[width=0.98\textwidth]{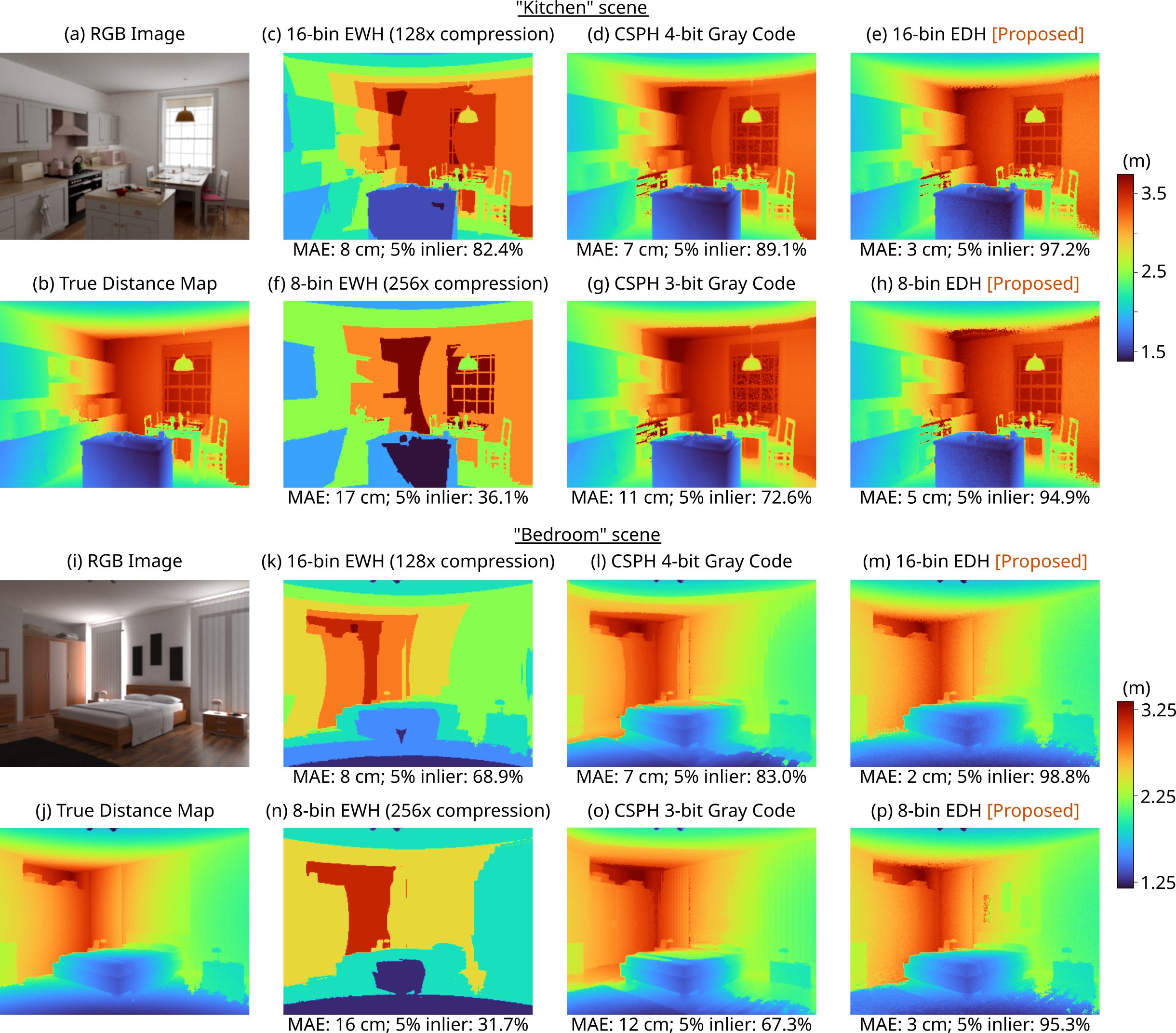}
  \caption{ {\bf Rendered results for ``kitchen'' and ``bedroom'' scenes:}
\rnote{(a,i) RGB images of rendered scenes.
(b,j) Ground-truth distance maps with (note different color scales for the two scenes).
(c,k) Distance maps reconstructed using a coarse EW histogram method that uses 16 equi-width bins ($128\times$ compression over a 2000 bin EW histogram) show strong quantization noise.
(d,l) Distance maps reconstructed using Gray code-based compressive histograms (CSPH) \cite{gutierrez2022compressive} provides improved visual quality.
(e,m) Our method provides reliable distance estimates with just 16-bin EDH, achieving $128\times$ compression over a conventional $2000$-bin EW histogram.
This approach uses a per-pixel ED histogrammer with the ``argmax'' distance estimator with no additional post-processing.
(f,n,g,o,h,p) Details in the distance maps are preserved even with an 8-bin EDH, achieving a $256\times$ compression over a full-resolution EW histogram and comparable distance map quality to the CSPH method.}
\label{fig:scene_results}}
\end{figure*}

\subsection{Single-Pixel Simulations}
For our single-pixel simulations we model the laser pulse shape as a Gaussian with two different FWHM values of $\SI{1}{\ns}$ and $\SI{5}{\ns}$ (which we refer to as ``narrow'' and ``wide'' pulses).\footnote{
Note that although there are picosecond FWHM lasers commercially available, they are costlier than nanosecond lasers.
Thus we deliberately chose ns-range pulse widths for our simulations to assess performance in low-cost, resource-constrained settings.}
This model is common in the literature for the time-varying light intensity seen from a pulsed-laser source.
We use a Poisson distribution to generate photon arrival delays from this time-varying Gaussian intensity profile.

\smallskip
\noindent{\bf Simulation parameters:}
Our baseline method is an EW histogram consisting of $1024$ bins and a bin-width of $\Delta=\SI{128}{\ps}$, which corresponds to a laser pulse rate of $\sim$$\SI{7.5}{\mega\hertz}$ (maximum distance $\approx \SI{20}{\m}$).
The laser strength is varied over $\Phi_\text{sig} \in \{0.1, 0.2, 0.5,1.0,2.0\}$ signal photons, on average, per laser cycle.
The constant offset introduced by the background strength is varied over $\Phi_\text{bkg}/B \in \{10^{-4},5\times 10^{-4}, 10^{-3}, 5\times 10^{-3}\}$
ambient photons per laser cycle per EW bin.
Our simulations are of a 16-bin (four-stage) EDH with a window size of 1024 units of \SI{128}{\ps} each.
For each combination of signal and background, we conduct a Monte Carlo session of 100 histogrammer runs (of 5000 laser cycles each), with the peak position chosen uniformly randomly over the $\SI{20}{\m}$ range.

We evaluate performance using three different error metrics: (a) $\delta=0.05$ metric (fraction of estimates that fall within $5\%$ of the ground truth); (b) $\delta=0.01$ metric (fraction that fall within $1\%$ of the ground truth); and (c) mean absolute error with respect to the ground truth.
For the EDH simulation, we estimate the peak location using both the argmax and quadratic curve fit methods from Sec.~\ref{sec:edh}.
For the EW baseline we use the bin position with maximum photon count as the distance estimator.

The boxplots in Fig.~\ref{fig:singlepixel_montecarlo} show the $5\%$ and $1\%$-error metrics using a wide (\SI{5}{\nano\second}) Gaussian pulse.
Each box plot represents results over 20 different combinations of signal strength and background light level.
The left-hand plots show results for a single readout per run; the right-hand plots are for the median estimate over five readouts closely spaced after the initial 5000 cycles in a run.
The EW histogram baseline that uses 1024 bins generally performs well, although at the cost of requiring at least an order of magnitude more data. 
See the supplement for additional boxplots for mean-absolute error metric that show similar trend as the $\delta$-metrics.

\subsection{Flash LiDAR Dataset}
We use the ``iToF2dToF'' dataset \cite{gutierrez2021itof2dtof}, which is generated using a transient rendering engine.
This dataset consists of rendered RGB images of various scenes together with corresponding transient distributions for each scene point captured by an ideal time-resolved camera.
The rendering technique uses flash illumination, where the entire field of view is flood-illuminated by the laser pulse.
This method allows the possibility of multiple reflections between scene points before the photons return to the camera.
The transient distributions in this dataset exhibit more complex shapes and multiple peaks, unlike the single Gaussian peak model used for the single-pixel evaluations.
(The supplement provides some examples of these transient distributions.)
Treating the ideal transient as the scene ``impulse response'', we convolve it with a Gaussian kernel ($\sigma=\SI{0.6}{\ns}$) to emulate a FWHM $=\SI{1}{\nano\second}$ laser pulse.

The sum of the ground truth transients along the time-axis for each pixel is used as a proxy albedo map.
We assume a mean signal strength and a mean background strength for the whole scene and weight the ground truth transients proportional to the albedo of each pixel.
The ground truth transients in this dataset consist of $B=2000$ samples over a maximum distance range of $\SI{20}{\meter}$.

We use these transients to generate photon arrival delays by sampling an inhomogeneous Poisson distribution at each discrete distance location over $N=5000$ laser cycles.
As a baseline comparison, we use a full resolution EW histogram with $2000$ bins, where each bin has a width of $\SI{33.33}{\pico\second}$.
\rnote{Recent hardware implementations use coarse binning \cite{gyongy2021directTutorial,gyongy2023direct} to reduce the amount of data that must be transferred off-sensor.}
We compress these histograms by $128\times$ and $256\times$ by coarsely binning them into $16$-bins and $8$-bin EW histograms.
We use the center of the peak bin in an EW histogram as a distance estimate.

Fig.~\ref{fig:scene_results} show results for two scenes assuming an average signal strength of $\Phi_\text{sig}=2.0$ and average background strength of $\Phi_\text{bkg}/B=0.0001$ for both scenes.
For both the 16-bin and 8-bin EDH results, we use the argmax method with no additional smoothing or post-processing.
Remarkably, the degradation in overall distance reconstruction quality is almost imperceptible with our EDH method, even for the 8-bin case.
In contrast, the coarse-binning compression method shows a clear degradation in distance estimates due to severe quantization noise.
Observe that in the ``kitchen'' scene, the EDH method shows distance estimation errors near the top corner (where the two walls and the ceiling meet) where the transient distribution contains multiple peaks.
In the ``bedroom'' scene, the low reflectivity scene points (black picture frames on the back wall) have larger errors in the 8-bin case.
Suppl.~Sec.~\ref{app:add-results} shows a table of all quantitative metrics for the 25 scenes in this dataset.


\rnote{Fig.~\ref{fig:scene_results} also shows simulated comparison with the Gray-coding-based compressive histogram (CSPH) method \cite{gutierrez2022compressive}.
Observe that the distance maps with the CSPH method are visually smoother while those with our method have a ``grainy'' appearance.
This artifact arises from the slight jitter in bin boundary locations estimated by our EDH method and could be addressed by additional spatial smoothing during post processing.
Still, our method performs better in terms of mean-absolute-error and 5\% inlier metrics than the CSPH method.
Additional comparisons are shown in Suppl. Fig.~\ref{fig:csph_comparison}.}
Note that the CSPH method requires a global shared memory to store the compression matrix, while our method only needs the pixel CV registers.

\begin{figure*}[!t]
\centering 
\includegraphics[width=0.95\textwidth]{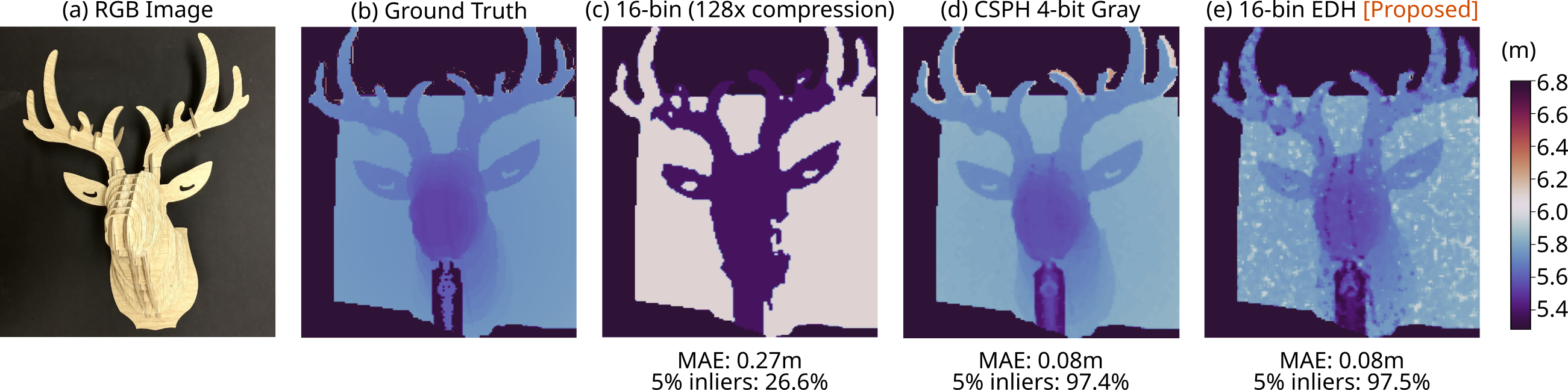}
\caption{\rnote{{\bf Experiment result with real-world ``elk'' dataset \cite{gupta2019asynchronous}}. (a) RGB image of the elk. (b) Ground truth distance map. (c) A conventional 16-bin EWH shows strong quantization artifacts. (d) The CSPH method provides visually smooth distance map reconstruction but suffers from some systematic errors. (e) Our 16-bin EDH method performs comparably to the CSPH method notwithstanding the grainier visual quality. (MAE=mean absolute error in meters, 5\% inliers=fraction of pixels within 5\% of the ground truth.)
}
\label{fig:expt_elk}}
\end{figure*}



\subsection{Hardware Emulation}
A hardware implementation of an in-pixel binner circuit based on race logic is currently not available.
Our simulation study includes the effect of Poisson noise, detector quantum efficiency, and low levels of dark-count noise, but not other sources of noise, such as afterpulsing.
Here we show hardware emulation results using real-world SPAD data that includes all realistic sources of noise, including dark counts, afterpulsing, and effects of dead-time.
We use raw photon timestamp data from a publicly available SPC dataset of experimentally captured photon-timestamp streams \cite{gupta2019asynchronous}.
This data was captured using a single-pixel SPAD-LiDAR prototype with a dead-time of \SI{100}{\ns}, dark count rate of $\approx 100$ counts/s and afterpulsing probability of $\approx 1\%$.
We emulate the behavior of both a 16-bin EWH and EDH by replaying the photon timestamp stream over \rnote{$5000$ laser cycles. Each stage of the EDH runs for $1250$ cycles before being frozen and running the next stage of the tree.}

A result is shown in Fig.~\ref{fig:expt_elk}. 
The final distance maps are denoised using a $3\times 3$ median filter.
The transient distribution in this dataset contains multiple peaks (due to interreflections).
The EDH result is generated by picking the first locally narrow bin location as the distance estimate.
The final result appears grainy because the method occasionally locks into the wrong peak.
Our method provides a better MAE and inlier metric compared to a 16-bin EW histogram result, which suffers from strong quantization errors.
\rnote{Our method also performs comparably to the Gray code-based compressive histograms method \cite{gutierrez2022compressive} in terms of the MAE and inlier metrics, although the visual quality of the CSPH result is much smoother. 
Additional results at varying SBR levels are shown in Suppl. Fig.~\ref{fig:elk_vary_sbr}.}

\section{Discussion and Future Directions}
We have proposed equi-depth histograms as a bandwidth- and energy-efficient representation for capturing scene distance information using SPCs.
Using a binner---which adaptively finds a given quantile---as our basic building block, we describe an equi-depth histogrammer that determines multiple bin boundaries without explicitly storing a history of photon counts.
The binner is amenable to implementation with race logic, a technology that operates in the delay domain, which is well suited to processing return events at a single-photon pixel.
Our approach shows promise for reducing bandwidth while maintaining similar distance accuracy as existing resource-hungry methods that rely on storing and processing equi-width histograms with thousands of bins.
An EDH-based SPC can achieve an energy savings of $\sim 10-100\times$, depending on various factors such as the number of laser pulses needed for convergence and energy consumption of each readout.
Suppl. Sec.~\ref{app:hardware} provides an energy-budget analysis.

Future work will experimentally evaluate real-world hardware implementations of this proposed method.
This work will require developing custom binner hardware that implements a high-speed race logic (PASS-INHIBIT) operation either in the analog domain or using an FPGA \cite{Tzimpragos_2021}.
Some speculative hardware designs (analog and digital) are shown in Suppl. Sec.~\ref{app:hardware}.

The ED histogram bin boundaries can be treated as a novel low-level scene representation.
Additional scene information (more than just scene distances) can be directly inferred from an 8- or 16-bin ED histogram.
\rnote{The binner circuit can be repurposed for intensity estimation by using the binner's register memory as a passive photon counter, or by artificially injecting fake pulses at a known rate to calibrate the absolute bin populations of the ED histogram.}
In the future we will investigate algorithms to directly learn scene properties such as reflectance, material properties, camera pose, and motion cues using ED histogram boundaries as the scene representation.

\section*{Acknowledgements}
This research was supported in part by NSF Award 2138471.

{\small
\bibliographystyle{IEEEtran}
\bibliography{reflist}

\begin{thebibliography}{10}
\providecommand{\url}[1]{#1}
\csname url@samestyle\endcsname
\providecommand{\newblock}{\relax}
\providecommand{\bibinfo}[2]{#2}
\providecommand{\BIBentrySTDinterwordspacing}{\spaceskip=0pt\relax}
\providecommand{\BIBentryALTinterwordstretchfactor}{4}
\providecommand{\BIBentryALTinterwordspacing}{\spaceskip=\fontdimen2\font plus
\BIBentryALTinterwordstretchfactor\fontdimen3\font minus
  \fontdimen4\font\relax}
\providecommand{\BIBforeignlanguage}[2]{{%
\expandafter\ifx\csname l@#1\endcsname\relax
\typeout{** WARNING: IEEEtran.bst: No hyphenation pattern has been}%
\typeout{** loaded for the language `#1'. Using the pattern for}%
\typeout{** the default language instead.}%
\else
\language=\csname l@#1\endcsname
\fi
#2}}
\providecommand{\BIBdecl}{\relax}
\BIBdecl

\bibitem{Rangwala2020}
S.~Rangwala, ``{The iPhone 12 - LiDAR At Your Fingertips},'' Online (accessed
  July 2, 2022), Nov 2020,
  \url{https://www.forbes.com/sites/sabbirrangwala/2020/11/12/the-iphone-12lidar-at-your-fingertips/?sh=548012a73e28}.

\bibitem{Ouster2022}
Ouster, ``{Fully autonomous turbine inspection with Clobotics and Ouster},''
  Online (accessed July 2, 2022), Jun 2022,
  \url{https://ouster.com/blog/fully-autonomous-turbine-inspection-with-clobotics-and-ouster/}.

\bibitem{Horiba}
Horiba, ``{FLIMera: Imaging camera for dynamic FLIM studies at real time video
  rates},'' Online (accessed July 2, 2022),
  \url{https://www.horiba.com/int/scientific/products/detail/action/show/Product/flimera-1989/}.

\bibitem{lange20003d}
R.~Lange, ``3d time-of-flight distance measurement with custom solid-state
  image sensors in cmos/ccd-technology,'' Ph.D. dissertation, University of
  Siegen, 2000, chapter 2.

\bibitem{Tzimpragos_2021}
G.~Tzimpragos, A.~Madhavan, D.~Vasudevan, D.~Strukov, and T.~Sherwood,
  ``In-sensor classification with boosted race trees,'' \emph{Communications of
  the {ACM}}, vol.~64, no.~6, pp. 99--105, jun 2021.

\bibitem{zhang2022first}
T.~Zhang, M.~J. White, A.~Dave, S.~Ghajari, A.~Raghuram, A.~C. Molnar, and
  A.~Veeraraghavan, ``First arrival differential lidar,'' in \emph{2022 IEEE
  International Conference on Computational Photography (ICCP)}.\hskip 1em plus
  0.5em minus 0.4em\relax IEEE, 2022, pp. 1--12.

\bibitem{white2022differential}
M.~White, S.~Ghajari, T.~Zhang, A.~Dave, A.~Veeraraghavan, and A.~Molnar, ``A
  differential spad array architecture in 0.18 $\mu$m cmos for hdr imaging,''
  in \emph{2022 IEEE International Symposium on Circuits and Systems
  (ISCAS)}.\hskip 1em plus 0.5em minus 0.4em\relax IEEE, 2022, pp. 292--296.

\bibitem{gyongy2021directTutorial}
I.~Gyongy, N.~A. Dutton, and R.~K. Henderson, ``Direct time-of-flight
  single-photon imaging,'' \emph{IEEE Transactions on Electron Devices}, 2021.

\bibitem{shin2016photon}
D.~Shin, F.~Xu, D.~Venkatraman, R.~Lussana, F.~Villa, F.~Zappa, V.~K. Goyal,
  F.~N. Wong, and J.~H. Shapiro, ``Photon-efficient imaging with a
  single-photon camera,'' \emph{Nature communications}, vol.~7, no.~1, pp.
  1--8, 2016.

\bibitem{rapp2021high}
J.~Rapp, Y.~Ma, R.~M. Dawson, and V.~K. Goyal, ``High-flux single-photon
  lidar,'' \emph{Optica}, vol.~8, no.~1, pp. 30--39, 2021.

\bibitem{gupta2019photon}
A.~Gupta, A.~Ingle, A.~Velten, and M.~Gupta, ``Photon-flooded single-photon 3d
  cameras,'' in \emph{Proceedings of the IEEE/CVF Conference on Computer Vision
  and Pattern Recognition}, 2019, pp. 6770--6779.

\bibitem{gupta2019asynchronous}
A.~Gupta, A.~Ingle, and M.~Gupta, ``Asynchronous single-photon 3d imaging,'' in
  \emph{Proceedings of the IEEE/CVF International Conference on Computer
  Vision}, 2019, pp. 7909--7918.

\bibitem{patanwala2021high}
S.~M. Patanwala, I.~Gyongy, H.~Mai, A.~A{\ss}mann, N.~A. Dutton, B.~R. Rae, and
  R.~K. Henderson, ``A high-throughput photon processing technique for range
  extension of spad-based lidar receivers,'' \emph{IEEE Open Journal of the
  Solid-State Circuits Society}, vol.~2, pp. 12--25, 2021.

\bibitem{sheehan2021sketching}
M.~P. Sheehan, J.~Tachella, and M.~E. Davies, ``A sketching framework for
  reduced data transfer in photon counting lidar,'' \emph{IEEE Transactions on
  Computational Imaging}, vol.~7, pp. 989--1004, 2021.

\bibitem{gutierrez2022compressive}
F.~Gutierrez-Barragan, A.~Ingle, T.~Seets, M.~Gupta, and A.~Velten,
  ``Compressive single-photon 3d cameras,'' in \emph{Proceedings of the
  IEEE/CVF Conference on Computer Vision and Pattern Recognition}, 2022, pp.
  17\,854--17\,864.

\bibitem{colacco2012compressive}
A.~Cola{\c{c}}o, A.~Kirmani, G.~A. Howland, J.~C. Howell, and V.~K. Goyal,
  ``Compressive depth map acquisition using a single photon-counting detector:
  Parametric signal processing meets sparsity,'' in \emph{2012 IEEE Conference
  on Computer Vision and Pattern Recognition}.\hskip 1em plus 0.5em minus
  0.4em\relax IEEE, 2012, pp. 96--102.

\bibitem{poisson20222}
V.~Poisson, W.~Guicquero, D.~Coriat, and G.~Sicard, ``A 2-stage em algorithm
  for online peak detection, an application to tcspc data,'' \emph{IEEE
  Transactions on Circuits and Systems II: Express Briefs}, 2022.

\bibitem{hutchings2019reconfigurable}
S.~W. Hutchings, N.~Johnston, I.~Gyongy, T.~Al~Abbas, N.~A. Dutton, M.~Tyler,
  S.~Chan, J.~Leach, and R.~K. Henderson, ``A reconfigurable 3-d-stacked spad
  imager with in-pixel histogramming for flash lidar or high-speed
  time-of-flight imaging,'' \emph{IEEE Journal of Solid-State Circuits},
  vol.~54, no.~11, pp. 2947--2956, 2019.

\bibitem{lindner2018252}
S.~Lindner, C.~Zhang, I.~M. Antolovic, M.~Wolf, and E.~Charbon, ``A 252$\times$
  144 spad pixel flash lidar with 1728 dual-clock 48.8 ps tdcs, integrated
  histogramming and 14.9-to-1 compression in 180nm cmos technology,'' in
  \emph{2018 IEEE Symposium on VLSI Circuits}.\hskip 1em plus 0.5em minus
  0.4em\relax IEEE, 2018, pp. 69--70.

\bibitem{vornicu2019tof}
I.~Vornicu, A.~Darie, R.~Carmona-Galan, and {\'A}.~Rodr{\'\i}guez-V{\'a}zquez,
  ``Tof estimation based on compressed real-time histogram builder for spad
  image sensors,'' in \emph{2019 IEEE International Symposium on Circuits and
  Systems (ISCAS)}.\hskip 1em plus 0.5em minus 0.4em\relax IEEE, 2019, pp.
  1--4.

\bibitem{po2021adaptive}
R.~Po, A.~Pediredla, and I.~Gkioulekas, ``Adaptive gating for single-photon 3d
  imaging,'' \emph{arXiv preprint arXiv:2111.15047}, 2021.

\bibitem{frames2016}
\BIBentryALTinterwordspacing
M.~Grossniklaus, D.~Maier, J.~Miller, S.~Moorthy, and K.~Tufte, ``Frames:
  Data-driven windows,'' in \emph{Proceedings of the 10th ACM International
  Conference on Distributed and Event-Based Systems}, ser. DEBS 16.\hskip 1em
  plus 0.5em minus 0.4em\relax New York, NY, USA: Association for Computing
  Machinery, 2016, p. 13–24. [Online]. Available:
  \url{https://doi.org/10.1145/2933267.2933304}
\BIBentrySTDinterwordspacing

\bibitem{segmenting2001}
E.~J. Keogh, S.~Chu, D.~Hart, and M.~J. Pazzani, ``An online algorithm for
  segmenting time series,'' in \emph{Proceedings of the 2001 IEEE International
  Conference on Data Mining}, ser. ICDM '01.\hskip 1em plus 0.5em minus
  0.4em\relax USA: IEEE Computer Society, 2001, p. 289–296.

\bibitem{shatkay1996approximate}
H.~Shatkay and S.~B. Zdonik, ``Approximate queries and representations for
  large data sequences,'' in \emph{Proceedings of the Twelfth International
  Conference on Data Engineering}.\hskip 1em plus 0.5em minus 0.4em\relax IEEE,
  1996, pp. 536--545.

\bibitem{himberg2001segmentation}
J.~Himberg, K.~Korpiaho, H.~Mannila, J.~Tikanmaki, and H.~Toivonen, ``Time
  series segmentation for context recognition in mobile devices,'' in
  \emph{Proceedings 2001 IEEE International Conference on Data Mining}, 2001,
  pp. 203--210.

\bibitem{scanstats}
J.~Glaz and S.~Wallenstein, \emph{Scan Statistics - Methods and
  Applications}.\hskip 1em plus 0.5em minus 0.4em\relax Springer, 01 2009.

\bibitem{preston-event-2009}
D.~Preston, P.~Protopapas, and C.~Brodley, ``Event discovery in time series,''
  in \emph{Proceedings of the 2009 SIAM International Conference on Data
  Mining}.\hskip 1em plus 0.5em minus 0.4em\relax SIAM, 2009, pp. 61--72.

\bibitem{guha-approx-2006}
\BIBentryALTinterwordspacing
S.~Guha, N.~Koudas, and K.~Shim, ``Approximation and streaming algorithms for
  histogram construction problems,'' \emph{ACM Trans. Database Syst.}, vol.~31,
  no.~1, p. 396–438, mar 2006. [Online]. Available:
  \url{https://doi.org/10.1145/1132863.1132873}
\BIBentrySTDinterwordspacing

\bibitem{manku-approx-2002}
G.~S. Manku and R.~Motwani, ``Approximate frequency counts over data streams,''
  in \emph{Proceedings of the 28th International Conference on Very Large Data
  Bases}, ser. VLDB '02.\hskip 1em plus 0.5em minus 0.4em\relax VLDB Endowment,
  2002, p. 346–357.

\bibitem{madhavan2015race}
A.~Madhavan, T.~Sherwood, and D.~Strukov, ``Race logic: abusing hardware race
  conditions to perform useful computation,'' \emph{IEEE Micro}, vol.~35,
  no.~3, pp. 48--57, 2015.

\bibitem{Smith2018}
J.~Smith, ``Space-time algebra: A model for neocortical computation,'' in
  \emph{2018 {ACM}/{IEEE} 45th Annual International Symposium on Computer
  Architecture ({ISCA})}.\hskip 1em plus 0.5em minus 0.4em\relax {IEEE}, Jun
  2018.

\bibitem{lee2002race}
S.-J. Lee and H.-J. Yoo, ``Race logic architecture (rala): a novel logic
  concept using the race scheme of input variables,'' \emph{IEEE Journal of
  Solid-State Circuits}, vol.~37, no.~2, pp. 191--201, 2002.

\bibitem{madhavan2014race}
A.~Madhavan, T.~Sherwood, and D.~Strukov, ``Race logic: A hardware acceleration
  for dynamic programming algorithms,'' \emph{ACM SIGARCH Computer Architecture
  News}, vol.~42, no.~3, pp. 517--528, 2014.

\bibitem{lindell2018single}
D.~B. Lindell, M.~O'Toole, and G.~Wetzstein, ``Single-photon 3d imaging with
  deep sensor fusion.'' \emph{ACM Trans. Graph.}, vol.~37, no.~4, pp. 113--1,
  2018.

\bibitem{sun2019simple}
F.~Sun, Y.~Xu, Z.~Wu, and J.~Zhang, ``A simple analytic modeling method for
  spad timing jitter prediction,'' \emph{IEEE Journal of the Electron Devices
  Society}, vol.~7, pp. 261--267, 2019.

\bibitem{shimada2021back}
S.~Shimada, Y.~Otake, S.~Yoshida, S.~Endo, R.~Nakamura, H.~Tsugawa, T.~Ogita,
  T.~Ogasahara, K.~Yokochi, Y.~Inoue \emph{et~al.}, ``A back illuminated 6
  $\mu$m spad pixel array with high pde and timing jitter performance,'' in
  \emph{2021 IEEE International Electron Devices Meeting (IEDM)}.\hskip 1em
  plus 0.5em minus 0.4em\relax IEEE, 2021, pp. 20--1.

\bibitem{guha2009stream}
S.~Guha and A.~McGregor, ``Stream order and order statistics: Quantile
  estimation in random-order streams,'' \emph{SIAM Journal on Computing},
  vol.~38, no.~5, pp. 2044--2059, 2009.

\bibitem{munro1980selection}
J.~I. Munro and M.~S. Paterson, ``Selection and sorting with limited storage,''
  \emph{Theoretical computer science}, vol.~12, no.~3, pp. 315--323, 1980.

\bibitem{coates1968correction}
P.~Coates, ``The correction for photonpile-up'in the measurement of radiative
  lifetimes,'' \emph{Journal of Physics E: Scientific Instruments}, vol.~1,
  no.~8, p. 878, 1968.

\bibitem{pediredla2018signal}
A.~K. Pediredla, A.~C. Sankaranarayanan, M.~Buttafava, A.~Tosi, and
  A.~Veeraraghavan, ``Signal processing based pile-up compensation for gated
  single-photon avalanche diodes,'' \emph{arXiv preprint arXiv:1806.07437},
  2018.

\bibitem{beer2018background}
M.~Beer, J.~F. Haase, J.~Ruskowski, and R.~Kokozinski, ``Background light
  rejection in {SPAD}-based {LiDAR} sensors by adaptive photon coincidence
  detection,'' \emph{Sensors}, vol.~18, no.~12, 2018.

\bibitem{gutierrez2021itof2dtof}
F.~Gutierrez-Barragan, H.~Chen, M.~Gupta, A.~Velten, and J.~Gu, ``itof2dtof: A
  robust and flexible representation for data-driven time-of-flight imaging,''
  \emph{IEEE Transactions on Computational Imaging}, vol.~7, pp. 1205--1214,
  2021.

\bibitem{gyongy2023direct}
I.~Gyongy, A.~T. Erdogan, N.~A. Dutton, G.~M. Mart{\'\i}n, A.~Gorman, H.~Mai,
  F.~M. Della~Rocca, and R.~K. Henderson, ``A direct time-of-flight image
  sensor with in-pixel surface detection and dynamicc vision,'' \emph{J. Sel.
  Top. Quant. Elect.}, 2023.

\bibitem{abramowitz_1964}
M.~Abramowitz and I.~A. Stegun, \emph{Handbook of Mathematical Functions: With
  Formulas, Graphs, and Mathematical Tables}, 9th~ed.\hskip 1em plus 0.5em
  minus 0.4em\relax Dover American Nurses Association Publications, 1964,
  vol.~55.

\bibitem{grimmett_probbook_2001}
G.~R. Grimmett and D.~R. Stirzaker, \emph{Probability and Random Processes},
  3rd~ed.\hskip 1em plus 0.5em minus 0.4em\relax Oxford University Press, 2001.

\bibitem{mitzenmacher2017probability}
M.~Mitzenmacher and E.~Upfal, \emph{Probability and computing: Randomization
  and probabilistic techniques in algorithms and data analysis}.\hskip 1em plus
  0.5em minus 0.4em\relax Cambridge university press, 2017.

\bibitem{morimoto2020megapixel}
K.~Morimoto, A.~Ardelean, M.-L. Wu, A.~C. Ulku, I.~M. Antolovic, C.~Bruschini,
  and E.~Charbon, ``Megapixel time-gated spad image sensor for 2d and 3d
  imaging applications,'' \emph{Optica}, vol.~7, no.~4, pp. 346--354, 2020.

\end{thebibliography}
}

\clearpage
\onecolumn
\normalsize
\pagebreak
\renewcommand{\figurename}{Supplementary Figure}
\renewcommand{\tablename}{Supplementary Table}
\renewcommand{\thesection}{S.\arabic{section}}
\renewcommand{\thesection}{S.\arabic{section}}
\renewcommand{\theequation}{S\arabic{equation}}
\setcounter{figure}{0}
\setcounter{section}{0}
\setcounter{equation}{0}
\setcounter{page}{1}

\begin{center}
\Large Supplementary Document for\\
\Large ``Count-Free Single-Photon 3D Imaging with Race Logic'' \\[0.2cm]
\large Atul Ingle, David Maier\\[0.1cm]
\large Portland State University\\[0.1cm]
\normalsize \texttt{\{ingle2, maier\}@pdx.edu}
\end{center}

    \section{Single-stage Binner: Markov-chain Analysis \label{app:markov}}
    
    We model the control value of a 2-bin median-tracking binner as a Markov chain (Suppl. Fig. \ref{fig:markov_chain}).
    For simplicity we consider a window length $L$ with a binner that only takes discrete integer valued steps of size $\pm 1$. 
    At each laser cycle $t$ in a run, the binner's  control value $E_t$ is a random variable that lives in the set $\mathcal{L}=\{0,\ldots,L\}$, where $L$ denotes the number of window locations.
    After each laser cycle the sensor pixel receives photons in different window locations $i$ with independent non-identically Poisson-distributed counts $N_i \sim \text{Poisson}(r_i)$, where $r_i$ is the mean of the Poisson observation at location $i$, $1\leq i \leq L$.\footnote{Note that in this model we assume that the binner chooses the boundaries between the discrete window locations, hence there are $L+1$ possible control value locations for $L$ window locations.
    With small modifications to indexing, the following analysis will still hold in the case where the binner output is the same length as the number of window locations.}
    
    \begin{figure}[!ht]
        \centering
       \includegraphics[width=0.9\textwidth]{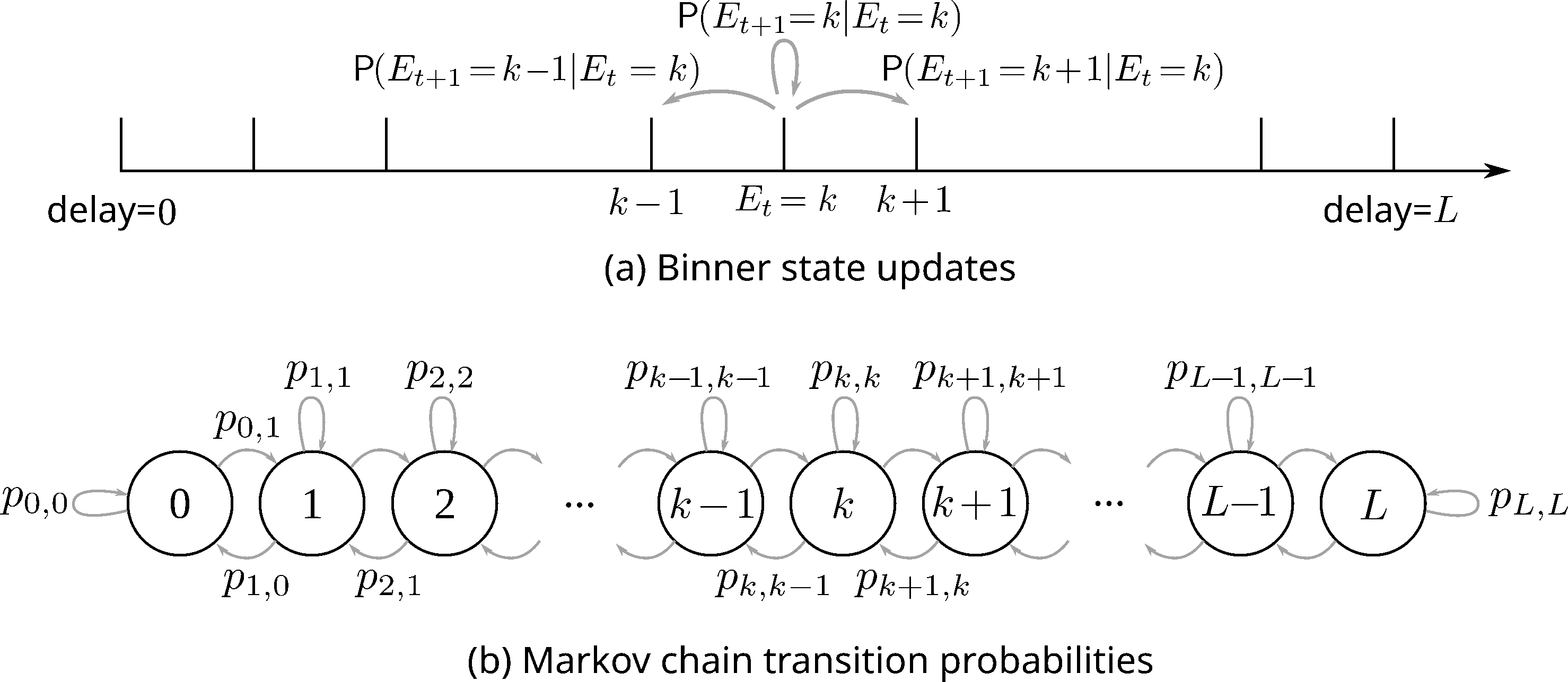}
        \caption{{\bf Markov chain model for online updates of a binner.}
          (a) After the $t^\text{th}$ laser cycle the binner's control value (denoted by $E_t$) is updated to a new value $E_{t+1}$ depending on the number of photons received between the delay ranges of $[0,E_t)$ and $[E_t, L)$.
          In the case of a step size of 1, the control value is decremented (incremented) by one if the total on the left is larger (smaller) than the total on the right.
          Otherwise it stays unchanged.
          (b) The binner's control value can be modeled using a Markov chain with transition probabilities $p_{i,j}$ where $0 \leq i, j \leq L.$
        These probabilities are time-homogeneous (i.e. they do not change over time).
        \label{fig:markov_chain}}
    \end{figure}
    
    In our race logic-based binner implementation, we do not directly observe the $N_i$'s.
    We track if the total photons on the left of $E_t$ is larger or smaller than the total on the right.
    On the $t^\text{th}$ laser cycle, our algorithm increments or decrements the estimate depending on the which side of current control value $E_t$ receives more photons:
    
    \[
    E_{t+1} =
    \left\{
    	\begin{array}{ll}
        \max(0, E_t - 1)  & \mbox{if } \sum_{i=0}^{E_t}N_i > \sum_{i=E_t}^L N_i \\
        \min(E_t + 1, L)  &  \mbox{if } \sum_{i=0}^{E_t}N_i < \sum_{i=E_t}^L N_i \\
        E_t      & \mbox{if } \sum_{i=0}^{E_t} N_i = \sum_{i=E_t}^L N_i,
    	\end{array}
    \right.
    \]
    where the max and min checks ensure that we do not overflow the limits $0$ and $L$.
    
    We will denote the transition probabilities by $p_{i,j} = \mathsf{P}\left( E_{t+1} = j | E_{t} = i\right).$
    For convenience, define $s_m \defeq \sum_0^{m-1} r_i$ for $1\leq m \leq L$ and $t_m = \sum_m^{L-1} r_i$ for $0 \leq m \leq L-1$, and note that $s_L = t_0 = \sum_0^{L-1} r_i$.
    For $1 \leq k \leq L-1$, the probability that the estimate moves to the right is given by:
    \begin{align*}
        p_{k,k+1} &= \mathsf{P}\left( \sum_{i=k}^{L-1} N_i > \sum_{i=0}^{k-1} N_i \right) \\
                &= 1 - S\left(0; \sum_{i=k}^{L-1}r_i, \sum_{i=0}^{k-1} r_i\right) \\
                &= 1 - S(0; t_k, s_k)
    \end{align*}
    where $S(\cdot; \mu_1, \mu_2)$ is the cumulative distribution function (cdf) of a Skellam distribution with parameters $\mu_1$ and $\mu_2$.
    The probability that the estimate moves to the left is given by:
    \begin{align*}
        p_{k,k-1} &= \mathsf{P}\left( \sum_{i=0}^{k-1} N_i > \sum_{i=k}^{L-1} N_i \right) \\
                &= 1 - S\left(0; \sum_{i=0}^{k-1}r_i, \sum_{i=k}^{L-1} r_i\right) \\
                &= 1 - S(0; s_k, t_k).
    \end{align*}
    Finally, the probability that the estimate does not move is given by
    $$
    p_{k,k} = S(0; t_k, s_k) + S(0; s_k, t_k) - 1 = e^{-s_L} I_0 \left(2\sqrt{t_k s_k}\right)
    $$
    where $I_0(\cdot)$ is the zeroth-order modified Bessel function of the second kind \cite{abramowitz_1964}.
    The edge cases $k=0$ and $k=L$ are handled separately:
    
    \begin{align*}
        p_{0,0} &= e^{-s_L} \\
        p_{0,1} &= 1 - e^{-s_L} \\
        p_{B,L-1} &= 1 - e^{-s_L} \\
        p_{L,L} &= e^{-s_L}.
    \end{align*}

    The right stochastic transition matrix of this Markov chain is a tri-diagonal:
    \begin{equation}
    \mathbf{P} =
    \begin{bmatrix}
    p_{0,0} & p_{0,1} &         &         &              &        &       \\
    p_{1,0} & p_{1,1} & p_{1,2} &         &              &        &       \\
            & p_{2,1} & p_{2,2} & p_{2,3} &              &        &       \\
            &         & p_{3,2} & p_{3,3} & p_{3,4}      &        &       \\
            &         &         & \ddots & \ddots & \ddots  &        &       \\
            &         &         &         & p_{L-2,L-1} & p_{L-1,L-1} & p_{L-1,L}  \\
            &         &         &         &     & p_{L,L-1} & p_{L,L}  \\
    \end{bmatrix}
    \label{eq:transition_matrix}
    \end{equation}
    Note that this Markov chain is irreducible and aperiodic.
    By the Perron-Frobenius theorem \cite{grimmett_probbook_2001}, it must have a stationary distribution.
    If we start with an initial state chosen uniformly randomly and allow the chain to run for a long time, we expect the chain to live close to the median of $r_i$.
    Note, however, that there is still a non-zero probability the Markov chain will make excursions around these expected settling locations; the larger the excursions, lower the probability.
    These settling locations, therefore, should be understood as the modes of the stationary distribution of the Markov chain after the chain runs for a long time.
    
    \begin{theorem*}
    Let the total photon rate (signal+background) for the transient distribution be $\mu \defeq \sum_{i=1}^L r_i$.
    Suppose the binner control value $\mathsf{CV} \neq$ true median and it splits the transient distribution into two fractions $f$ and $1-f$ where $0<f<1$.
    Let $\epsilon >0$ be an arbitrarily small probability threshold.
    Then, as long as $\mu > \left(\frac{1}{\sqrt{f}-\sqrt{1-f}}\right)^2 \log\left(\frac{1}{\epsilon}\right),$
    the probability that $\mathsf{CV}$ does not move towards the median on the next cycle is $< \epsilon$.
    \end{theorem*}
    
    \begin{proof}
    Without loss of generality, suppose that the current control value $\mathsf{CV} < $ true median. 
    This case implies that
    $$
    \mu_1=\sum_{i=0}^\mathsf{CV} r_i < \sum_{i=\mathsf{CV}+1}^L r_i =\mu_2.
    $$
    Let the total photon rate (signal+background) be denoted by $\mu = \mu_1 + \mu_2$ and $f=\mu_1/\mu$ be the fraction of the total photon flux to the left of $\mathsf{CV}$.
    Since  $\mu_1 < \mu_2$, $f < 0.5$.
    At the subsequent laser cycle, the number of photons $N_1$ and $N_2$ in the early and late streams respectively are independent Poisson random variables with rates $\mu_1 = f\mu$ and $\mu_2=(1-f)\mu.$
    The probability that $\mathsf{CV}$ does not move closer to the true median on the next laser cycle is equal to the probability that the early stream contains at least as many photons as the late stream ($N_1 \geq N_2)$.
    This probability can be bounded above using a Chernoff bound \cite{mitzenmacher2017probability}:
    $$
        P(N_1 \geq N_2) \leq e^{-(\sqrt{\mu_1} - \sqrt{\mu_2})^2}  = e^{-(\sqrt{f\mu} - \sqrt{(1-f)\mu})^2}.
    $$
    Observe that $e^{-(\sqrt{f\mu} - \sqrt{(1-f)\mu})^2} < \epsilon$
    provided $\mu > \left(\frac{1}{\sqrt{f}-\sqrt{1-f}}\right)^2 \log\left(\frac{1}{\epsilon}\right),$ which completes the proof. 
    \end{proof}
    
    For example, suppose that the current CV splits the transient distributions into two segments in the ratio $1:9$, i.e., $f=0.1$.
    Then as long as the total photon rate $\mu > 9.8$, the probability that CV does not move closer to the true median in the next laser cycle is less than $\epsilon=2\%$.
    Note that this photon rate is much higher than what we would expect to see in a single laser cycle in a real SPAD-based 3D camera.
    This observation suggests adding a low-bit-depth memory element to each pixel that accumulates photons and updates CV every few ($\sim 10$) laser cycles.
    We find that the bound provided by this theoretical result is quite loose in practice, and there is a much higher probability of moving towards the true median even for lower values of $\mu$.
    
    It is instructive to study the behavior of the Markov chain model in simulation for some simple cases.
    Setting the window length $L=1000$ and step size of $\pm 1$, we can numerically compute the stationary distributions over a wide range of signal strengths, SBR and scene distances.
    We first compute the one-step transition matrix (Eq.~\ref{eq:transition_matrix}). 
    Since this matrix is right stochastic, the stationary distribution is the left (row) eigenvector corresponding to an eigenvalue of $1$.
    Suppl. Fig.~\ref{fig:suppl_statn_dist_examples} shows example stationary distributions for distances (true peak locations) in $\{100, 250, 400\}$, signal strengths in $\{0.1, 0.5, 1.0\}$ and SBRs in $\{ 0.01, 0.2, 0.5, 1.0\}.$
    We make the following key observations:
    \begin{itemize}
        \item The true median does not necessarily lie at the true peak location; it is pulled towards the overall midpoint for low signal and high background cases.
        \item The peaks of the stationary distributions are always aligned with the true median locations, implying that the most likely location of the CV is at the true median.
        \item The spread of the stationary distribution around the median is distance dependent; the spread is smaller when the SBR is high or when the true peak location is closer to the overall midpoint ($L/2=500$) of the range.
    \end{itemize}
    
    The spread of the stationary distribution is related to how far a binner will wander away from the true median once it converges.
    A smaller spread is desirable.
    Suppl. Fig.~\ref{fig:suppl_statn_dist_stats} shows the probability mass around small neighborhoods of $\pm 5, \pm 10, \pm 20$ from the true median.
    Higher probabilities are desirable because they indicate tighter concentration around the true median.
    These probabilities depend on the signal strength, SBR and scene distance (peak location).
    The higher the signal and SBR, the higher the probability of finding the binner's control value close to the true median.
    The largest spread is seen for intermediate SBR levels (second row): in this regime, the signal and background exert similar pulls on the control value, the former towards the peak location while the latter towards the overall midpoint of $500$.
    The effect of this spread can be mitigated by averaging multiple readouts.
    By the central limit theorem, we expect the spread to reduce by a factor of $\sqrt{M}$ if $M$ readouts are averaged.
    
    \begin{figure}
        \centering
        \includegraphics[width=0.9\textwidth]{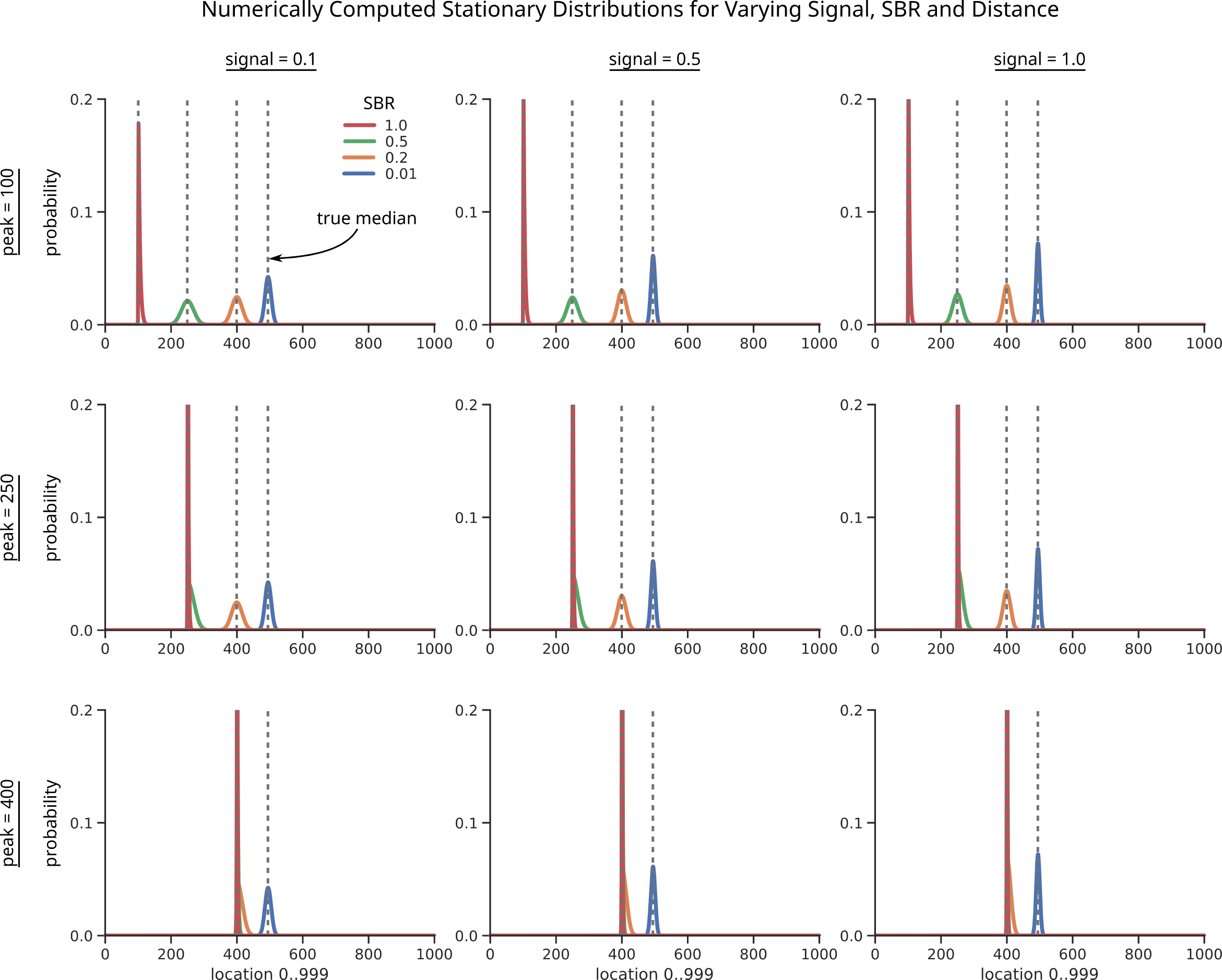}
        \caption{{\bf Stationary distributions of the Markov chain for different combinations of signal, SBR and scene distances.} These stationary distributions are computed numerically for a window length $L=1000$ and a range of signal, SBR and peak locations.
        The true median location does not necessarily line up with the true peak location; in case of weak signal or low SBR, the median is pulled away from the peak towards the overall midpoint of the range ($L/2 =  500$).
        Observe that in all cases, the true median is always located at the location of maximum probability (mode) of the stationary distributions.}
        \label{fig:suppl_statn_dist_examples}
    \end{figure}

    \begin{figure}
        \centering
        \includegraphics[width=0.9\textwidth]{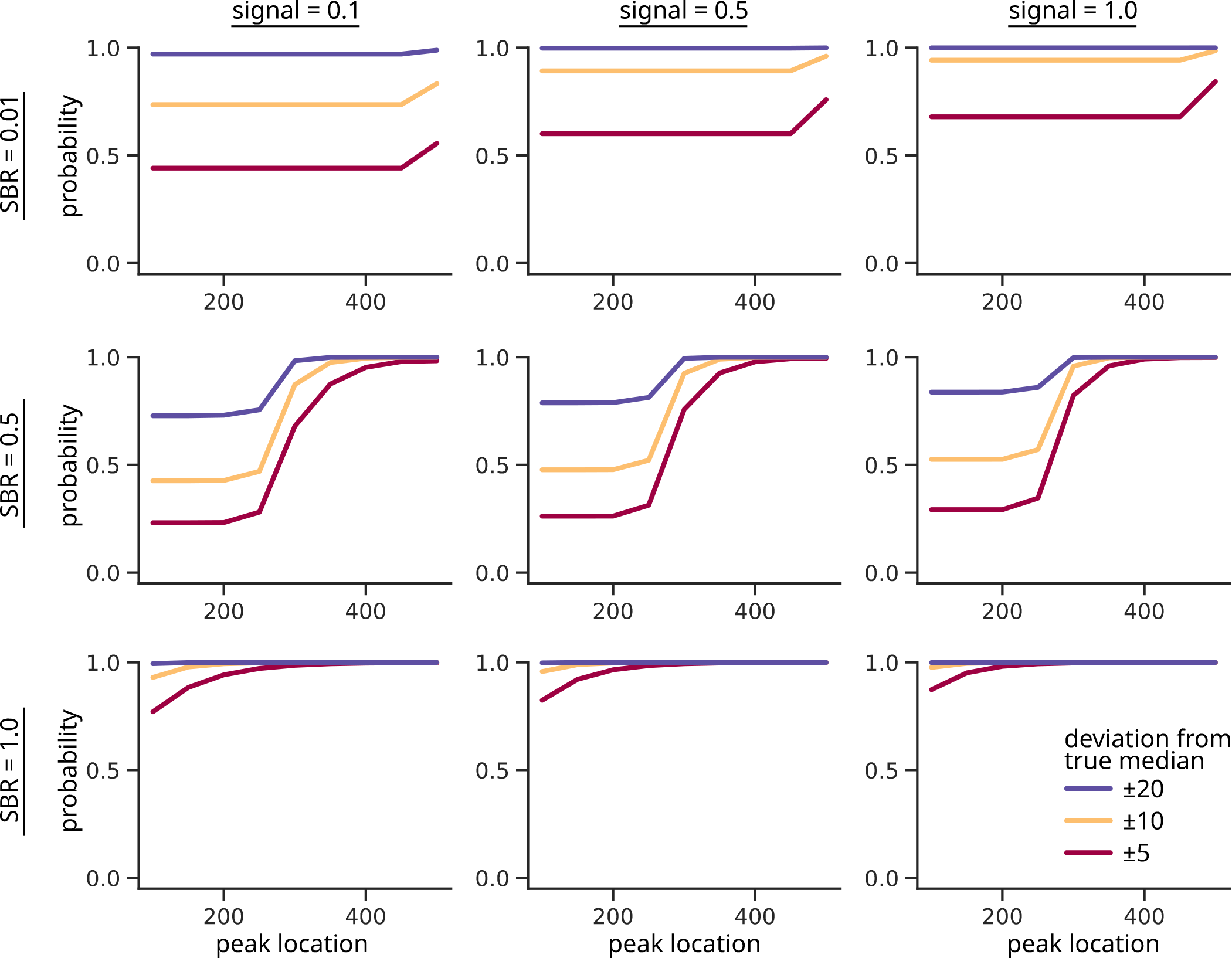}
        \caption{{\bf Concentration of the stationary distributions around the true median location.} 
        We study the spread of the stationary distribution around the true median location as a function of signal strength, SBR, and the true peak location (distance) for a window length of $1000$. At low SBRs, the binner control value may wander farther from the true median, but the probability of wandering farther than $\pm 20$ is small.
        At high SBRs, the binner is within $\pm 5$ of the true median with probability $>75\%$.
        At intermediate SBRs, both the peak location and the overall midpoint of the range exert similar pulls, hence the probabilities have a stronger dependence on the true peak location.}
        \label{fig:suppl_statn_dist_stats}
    \end{figure}
    
    \smallskip
    \noindent{\bf Extensions:} 
    Although the analysis show above was for a single-stage 2-bin binner, the theory can be applied to arbitrary binners in a tree-based EDH by restricting the window to a sub-range dictated by the parent binner.
    Our model assumes that the step sizes are $\pm 1$.
    We can derive transition matrices for larger step sizes.
    For example, the $\pm 10$-step-size case will still be a tri-diagonal matrix with non-zero entries along the main diagonal and the $10^\text{th}$ off-diagonals.
    In the next section we study some heuristics to speed up binner convergence to the median by using larger step sizes.

\clearpage
\section{Heuristics for Speeding up Binner Convergence \label{app:convergence}}
\begin{figure}[!htb]
    \centering
     \includegraphics[width=0.7\textwidth]{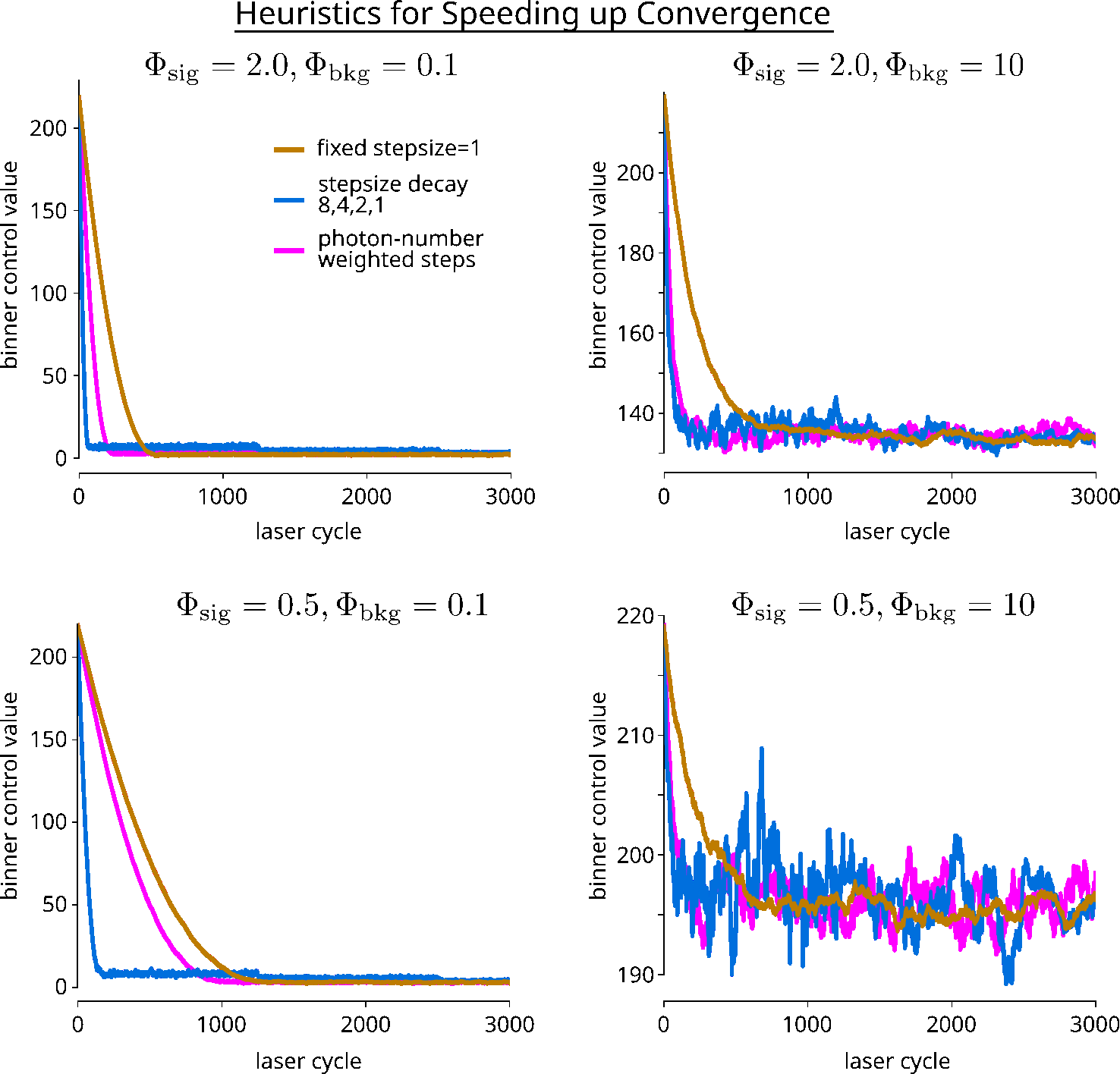}
     \caption{\textbf{Simulation study of rate of convergence.} We analyze three different stepping schemes: a n\"aive method that takes small constant steps of size 1, a weighted step method that uses steps equal to the difference in the number of return events in a cycle, and finally an ad hoc stepping schedule that takes large steps initially and then reduces the step size for subsequent laser pulses.
     A variable step-size schedule gives the best rate of convergence.
     (a) In a high-signal, low-background regime the variable step-size method quickly converges to the median.
     (b) In a high-signal, high-background regime we observe quick convergence but the final estimate is still quite noisy.
     (c) In a low-signal, low-background regime, the convergence takes longer than with a high-strength signal, but a variable step-size schedule achieves 10$\times$ improvement over the other schemes.
     (d) In the low-signal, high-background regime, the final estimates show large excursions from the true median location.
     These plots also suggest that there is an advantage in averaging multiple measurements post-convergence, especially in high-background situations.
     \label{fig:convergence_simulation}
     }
\end{figure}

We performed a Monte Carlo simulation study of a single binner with three different stepping schemes over different operating conditions (low and high signal power in the presence of low and high background light levels) for $5000$ laser cycles:
\begin{itemize}
    \item Constant small step size of 1,
    \item Photon-number-weighted step size: the step size is the difference between the number of early and late events,
    \item 4-stage coarse-to-fine step-size schedule ($8\rightarrow 4 \rightarrow 2 \rightarrow 1$), each for one-quarter of the total cycles.

\end{itemize}
Suppl. Fig.~\ref{fig:convergence_simulation} shows the evolution of a single binner's control value as a function of laser cycles.
(Although we ran the simulation to $5000$ cycles, these plots are only shown up to $3000$ cycles to focus on the interesting regions of the convergence trends.)
The coarse-to-fine stepping scheme gives the fastest convergence in all operating conditions.
It is faster than the constant-step-size method by a factor of at least $10\times$ in most SBR regimes.
In the high-signal-strength regime, the scheme settles to an optimum quite rapidly but then it is limited by the step size.
This effect is visible in Suppl. Fig.~\ref{fig:convergence_simulation}(a) and (c) where the blue plot seems to converge quite rapidly at first but then it makes further improvement with a finer step size as seen from the discrete jumps at 1250 and 2500 cycles. 
The results suggests that signal-dependent step-size optimizations can further speed up convergence.
In high-strength regimes, it helps to rapidly decay the step sizes to the finest level, whereas in low-strength regimes, larger steps should persist for larger fractions of the total exposure time budget.
How do we choose step sizes in practice?
We can use a heuristic step-size schedule informed by this simulation study.
We can perform an initial ``calibration'' scan to assess SBR conditions at different scene points.
High-SBR pixels can rapidly decay to small step sizes, while low-SBR pixels use large step sizes for longer durations.

The difference between the constant-small-step-size and weighted schemes is less marked, though the weighted scheme (magenta) does converge faster than the constant-size scheme (gold).
In Suppl. Fig.~\ref{fig:convergence_simulation} we only see a slight advantage for the low-strength regimes where the weighted scheme has similar convergence rate as the constant-step scheme.
This behavior arises because in the low signal plus background regimes, very few photon events are generated in each cycle, so there are very few cases where the boundary moves more than one unit per cycle.
In the high-strength regime, the weighted scheme does noticeably better than the constant-step scheme.
However, with a strong signal, convergence is fast with all schemes, although the variable-step-scheme does wander more around the median until it gets to a small step size.

\clearpage
\section{Additional Results}\label{app:add-results}
\subsection{Single-pixel simulations}

\begin{figure}[!ht]
  \centering
  \includegraphics[width=0.8\textwidth]{figures/SFig_mae_metrics.png}
  \caption{ {\bf Results of mean absolute error.} Each box plot represents the mean absolute error in distance estimates for 20 combinations of signal and background strengths (see main text). (a) Error plots with a single measurement from each binner. (b) Results after computing the median of distance estimates obtained from a series of 5 EDH read-outs.
  \label{fig:mae}
  }
\end{figure}

We record the mean absolute error (MAE) in the distance estimate where the mean is computed over $100$ runs and each run has a different randomly chosen ground-truth distance value.
The parameter combinations are the same as those for the single-pixel simulations described in the main text.
The absolute errors can be quite large (observe the logarithmic scale on the y-axis of these plots).
Suppl. Fig.~\ref{fig:mae} shows results for a wide (5 ns) Gaussian pulse.

When using a single EDH readout for estimating distance, the quadratic curve fitting method has a wider spread in MAE values as seen in Suppl. Fig.~\ref{fig:mae}(a). We posit this behavior arises from the fluctuations in the binner outputs.
If they do not precisely line up around the peak location, the curve fitting can ``amplify'' the error further.
This effect can be avoided by averaging multiple measurements.

Suppl. Fig.~\ref{fig:mae}(b) shows the effect of averaging multiple readouts on the mean absolute error of the final distance estimates.
With averaging, the quadratic fitting method achieves almost the same MAE as the conventional EWH method with a narrow Gaussian pulse.
Looking at the results in the figures, we believe that combining read-outs improves expected accuracy in all conditions, but especially for quadratic curve fitting.
Multiple read-outs do increase the bandwidth per run five-fold; nevertheless, that is still well less than a tenth of the data for a single 1024-bin EWH read-out.

\subsection{Example EDH Outputs for Rendered Transient Distributions}
Simulated transient-distribution data generated with a rendering engine \cite{gutierrez2021itof2dtof} give us some insights into the diversity of transient shapes that one may encounter in analysis and simulation.
A vast majority of the simulated transients consist of a single peak.
However, there are situations where (due to inter-reflections and multi-path effects) transient shapes can deviate significantly from the simple single-Gaussian-peak model that is widely used in practice.
We show some simulated plots of EDH outputs in Suppl. Fig.~\ref{fig:transient_examples}.
Note that this sample is biased in that we choose to show transients that we deemed ``interesting'' and challenging for an EDH because they have unusual shapes, including asymmetrical peaks, overlapping peaks, or long tails.

\begin{figure}[!ht]
  \centering
  \includegraphics[width=0.8\textwidth]{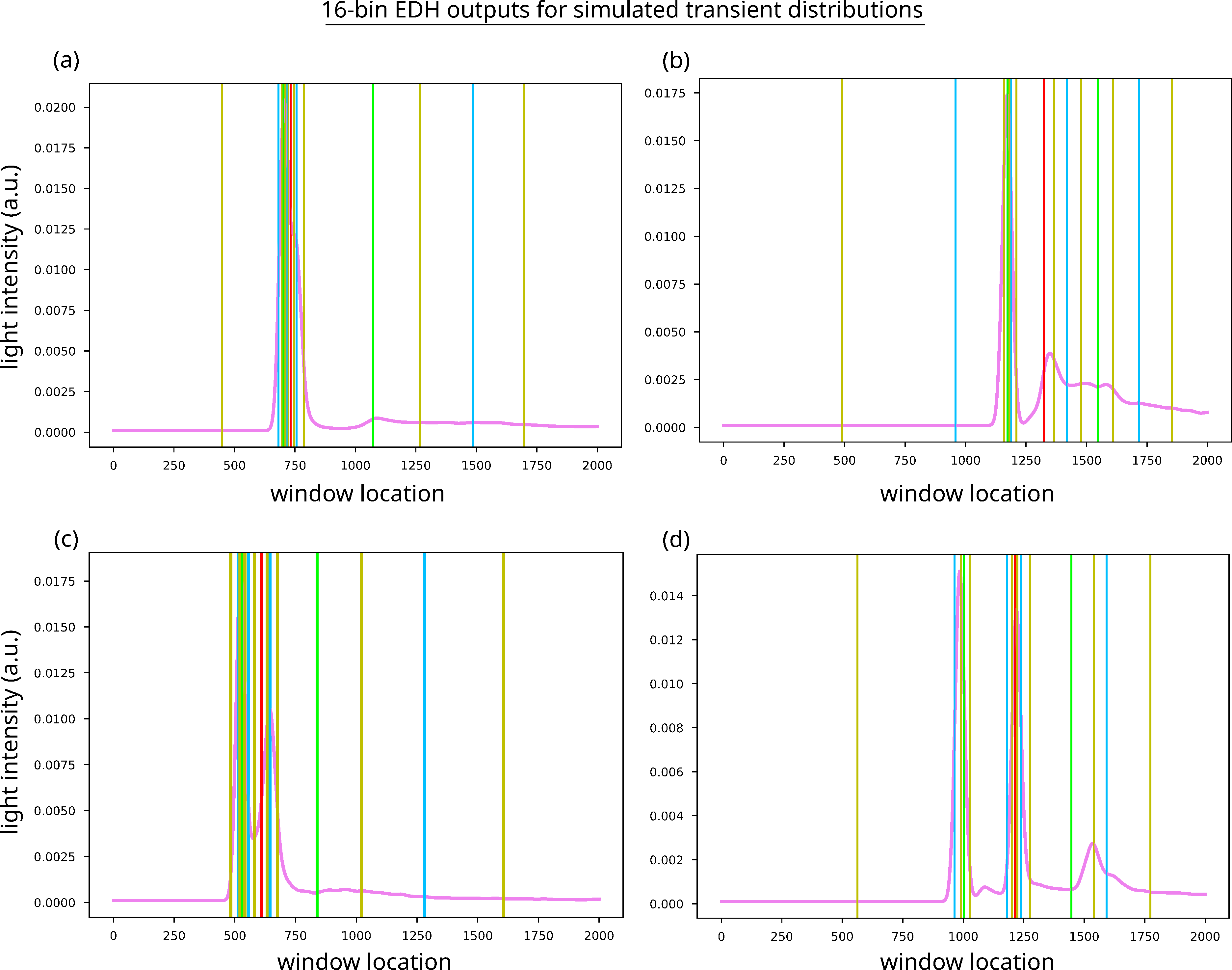}
  \caption{ {\bf EDH outputs for complex transient distributions.} (a) In case of a single strong peak, most of the bins cluster around the main peak. (b) In case of a long trailing edge, the last few bins have monotonically increasing bin widths, which capture the overall shape. (c) In case of two peaks, the locally narrow bins do a good job of clustering around the peak locations. (d) In case of more than 2 peaks, additional EDH stages may be needed to capture additional weak peaks such as the last peak shown here. 
  \label{fig:transient_examples}
  }
\end{figure}

A 16-bin EDH performs well in when there is a single distinct peak (Suppl. Fig.~\ref{fig:transient_examples}(a)).
Observe that around 6 side-bins of the 16 ED histogram bins are spent absorbing the background light while the remaining 10 bins are all quite narrow and clustered around the peak.
In cases with a peak followed by a long trailing edge (Suppl. Fig.~\ref{fig:transient_examples}(b)) there is a gradual increase in the widths of the ED histogram bins. 
In case of multiple peaks (Suppl. Fig.~\ref{fig:transient_examples}(c,d))
the locally narrow bins are clustered around the strong peaks.

\subsection{Additional Scene Results}
Suppl.~Fig.~\ref{fig:additional_scene_results} shows additional distance map reconstructions for three scenes. Conventional low-bin-count coarse EW histogramming methods fail due to heavy quantization artifacts. Our EDH method recovers fine details even with as few as 8 ED bins (over $200\times$ data compression). Our method also outperforms Gray-coding-based compressive histograms method \cite{gutierrez2022compressive} both in terms of mean absolute error and 5\% inlier metrics as seen in Suppl. Tab. \ref{tab:metrics_mae} and \ref{tab:metrics_5pc}.

\begin{figure}[!ht]
  \centering
  \includegraphics[width=0.98\textwidth]{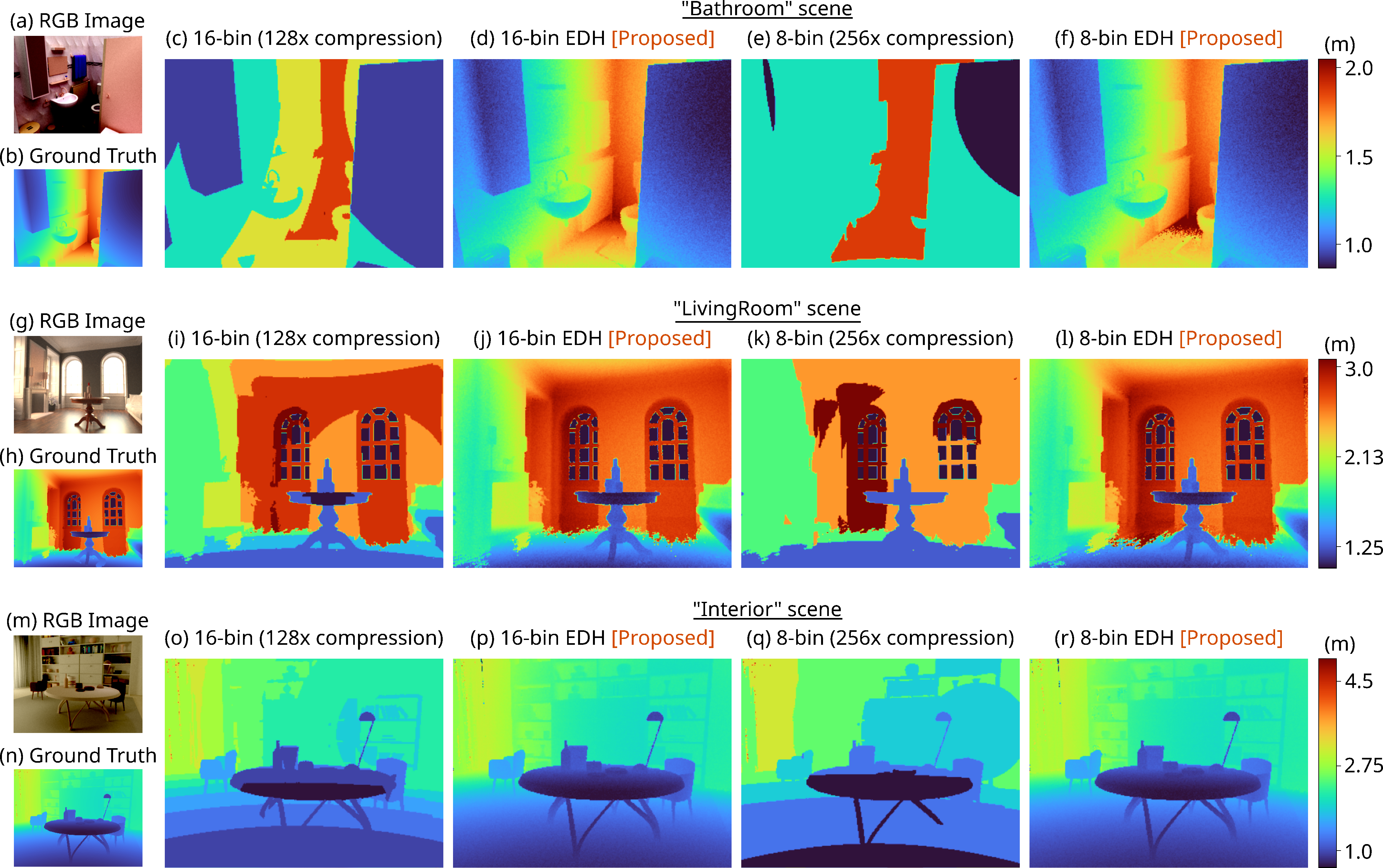}
  \caption{ {\bf Additional simulated distance map results.}
  (a--f) In the ``bathroom'' scene, our method preserves details such as the shape of the sink) even with just an 8-bin EDH. Heavy quantization noise makes it difficult to discern any structures in the distance maps from coarse 16-bin and 8-bin EW histograms. (g--l)  In the ``living room'' scene, conventional coarse-bin EWH does not cause severe loss of depth details because the coarse depth bins happen to line up close to the true depth values. However, notice that our EDH method recovers finer details such as the arch around the windows and the feet of the table. (m--r) In this ``interior'' scene, our method still recovers finer depth details in the background shelf that are heavily quantized in the coarse EWH reconstructions.
Note that these scenes have different ranges of minimum and maximum scene distances as denoted by their respective color bars.
    \label{fig:additional_scene_results}
  }
\end{figure}

\begin{figure}[!t]
\centering 
\includegraphics[width=0.5\columnwidth]{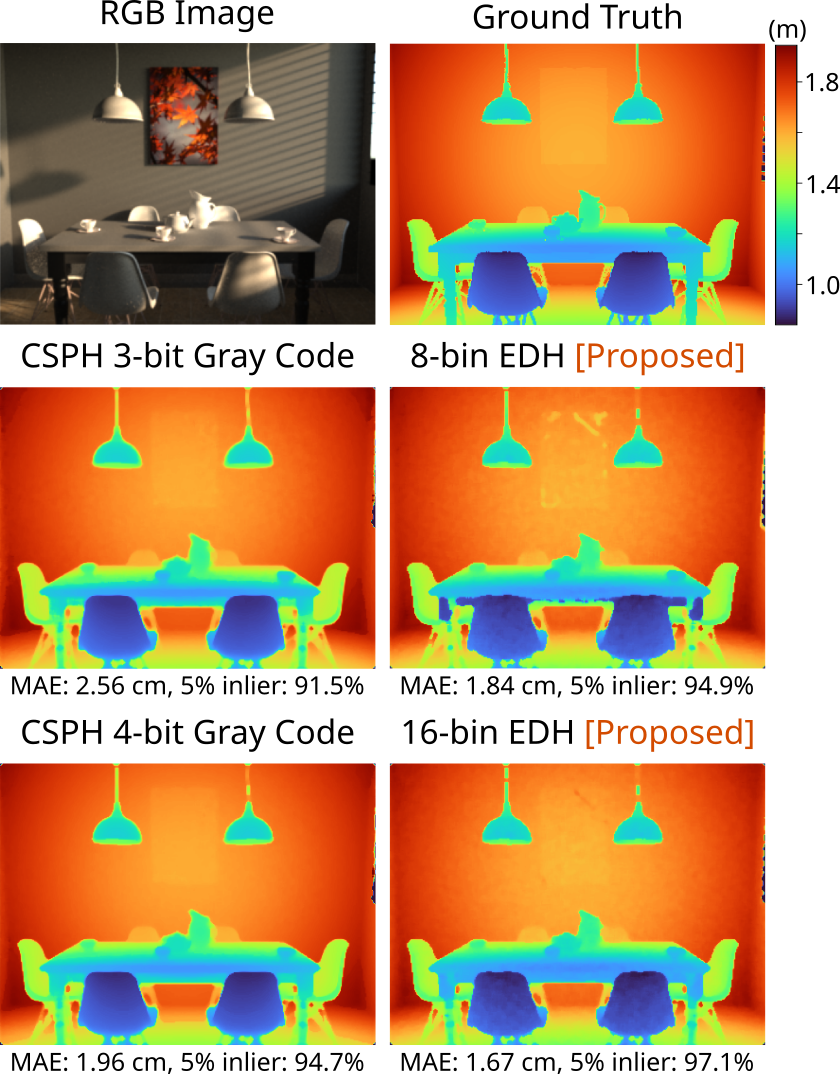}
\caption{{\bf Additional comparisons with compressive histograms \cite{gutierrez2022compressive}.} Our method gives comparable quality as the recent Gray-code-based CSPH method without the need for additional global memory for the compression matrix.
\label{fig:csph_comparison}}
\end{figure}

\clearpage

\begin{table}[!ht]
\addtolength{\tabcolsep}{-3pt}
\caption{Quantitative metrics for simulated scenes show large improvements in MAE with our method.
\label{tab:metrics_mae}
}
\centering
\begin{tabular}{@{}ccccccc@{}}
\toprule
\multicolumn{1}{l}{\multirow{2}{*}{scene name}} & \multicolumn{6}{c}{mean absolute error (cm)}                                                                                                                                                                                                                                                                        \\ \cmidrule(l){2-7} 
\multicolumn{1}{l}{}                            & \multicolumn{1}{l}{EWH 8 bin} & \multicolumn{1}{l}{EWH 16 bin} & \multicolumn{1}{l}{CSPH 3 bit} & \multicolumn{1}{l}{CSPH 4 bit} & \multicolumn{1}{l}{\begin{tabular}[c]{@{}l@{}}EDH 8 bin\\ {[}Proposed{]}\end{tabular}} & \multicolumn{1}{l}{\begin{tabular}[c]{@{}l@{}}EDH 16 bin\\ {[}Proposed{]}\end{tabular}} \\ \midrule
kitchen                                         & 17.01                         & 7.94                           & 10.87                          & 6.77                           & 4.16                                                                                   & 2.76                                                                                    \\
kitchen-2                                       & 12.66                         & 7.60                           & 1.82                           & 1.24                           & 0.81                                                                                   & 1.01                                                                                    \\
bedroom                                         & 16.32                         & 7.98                           & 12.18                          & 7.08                           & 2.77                                                                                   & 1.47                                                                                    \\
veach-ajar                                      & 16.71                         & 7.32                           & 17.78                          & 10.24                          & 3.99                                                                                   & 1.40                                                                                    \\
living-room                                     & 15.21                         & 7.26                           & 716.01                         & 714.64                         & 4.70                                                                                   & 3.37                                                                                    \\
office                                          & 16.04                         & 6.18                           & 63.74                          & 59.36                          & 3.60                                                                                   & 2.38                                                                                    \\
interior-scene                                  & 19.22                         & 7.47                           & 4.91                           & 3.48                           & 2.25                                                                                   & 2.04                                                                                    \\
bathroom-cycles-2                               & 19.26                         & 6.80                           & 7.22                           & 3.70                           & 1.19                                                                                   & 0.90                                                                                    \\
bathroom2                                       & 14.12                         & 8.68                           & 14.03                          & 7.09                           & 4.13                                                                                   & 2.68                                                                                    \\
bathroom                                        & 14.99                         & 7.56                           & 5.75                           & 3.56                           & 1.91                                                                                   & 1.52                                                                                    \\
staircase2                                      & 10.29                         & 4.95                           & 191.12                         & 189.12                         & 5.57                                                                                   & 4.98                                                                                    \\
hot-living                                      & 14.21                         & 7.34                           & 43.13                          & 39.12                          & 3.50                                                                                   & 2.41                                                                                    \\
staircase                                       & 18.00                         & 6.65                           & 46.50                          & 44.71                          & 7.25                                                                                   & 6.81                                                                                    \\
the-sitting-room                                & 15.76                         & 7.71                           & 19.19                          & 14.85                          & 3.20                                                                                   & 2.02                                                                                    \\
dining-room                                     & 16.11                         & 8.55                           & 2.56                           & 1.96                           & 1.85                                                                                   & 1.68                                                                                    \\
classroom                                       & 11.04                         & 5.94                           & 244.53                         & 241.65                         & 6.69                                                                                   & 6.11                                                                                    \\
cbox                                            & 16.14                         & 9.44                           & 3.55                           & 2.09                           & 2.73                                                                                   & 1.48                                                                                    \\
cycles-modern-kitchen-scene                     & 14.58                         & 7.74                           & 14.44                          & 8.25                           & 3.61                                                                                   & 2.35                                                                                    \\
living-room-3                                   & 14.15                         & 8.82                           & 14.24                          & 9.49                           & 3.32                                                                                   & 1.97                                                                                    \\
living-room-2                                   & 14.36                         & 6.55                           & 7.80                           & 4.30                           & 3.23                                                                                   & 1.92                                                                                    \\
my-office                                       & 15.74                         & 7.98                           & 17.20                          & 10.82                          & 4.63                                                                                   & 2.95                                                                                    \\
veach-bidir                                     & 14.24                         & 9.14                           & 11.37                          & 5.46                           & 1.24                                                                                   & 0.87                                                                                    \\
vgroove                                         & 28.30                         & 2.95                           & 1.84                           & 1.71                           & 0.57                                                                                   & 0.81                                                                                    \\
living-room-4                                   & 19.95                         & 7.56                           & 11.35                          & 6.86                           & 1.12                                                                                   & 0.79                                                                                    \\
breakfast-hall                                  & 6.53                          & 3.06                           & 540.80                         & 538.37                         & 1.97                                                                                   & 1.69                                                                                    \\ \bottomrule
\end{tabular}
\end{table}

\begin{table}[!ht]
\addtolength{\tabcolsep}{-3pt}
\caption{Quantitative metrics for simulated scenes show large improvements in 5\% inlier metric with our method.
\label{tab:metrics_5pc}
}
\centering
\begin{tabular}{@{}ccccccc@{}}
\toprule
\multirow{2}{*}{scene name} & \multicolumn{6}{c}{\% pixels within 5\% of true distance}                                                                                                                                   \\ \cmidrule(l){2-7} 
                            & EWH 8 bin & EWH 16 bin & CSPH 3 bit & CSPH 4 bit & \begin{tabular}[c]{@{}c@{}}EDH 8 bin\\ {[}Proposed{]}\end{tabular} & \begin{tabular}[c]{@{}c@{}}EDH 16 bin\\ {[}Proposed{]}\end{tabular} \\ \midrule
kitchen                     & 36.1      & 82.4       & 71.9       & 88.8       & 94.9                                                               & 97.2                                                                \\
kitchen-2                   & 17.2      & 25.0       & 84.7       & 95.4       & 98.6                                                               & 99.1                                                                \\
bedroom                     & 31.7      & 68.9       & 66.9       & 82.9       & 95.3                                                               & 98.8                                                                \\
veach-ajar                  & 44.9      & 93.1       & 56.9       & 85.7       & 95.1                                                               & 99.2                                                                \\
living-room                 & 34.0      & 72.9       & 73.3       & 77.3       & 91.8                                                               & 93.4                                                                \\
office                      & 26.3      & 66.1       & 56.9       & 78.9       & 84.7                                                               & 85.9                                                                \\
interior-scene              & 15.8      & 56.1       & 92.1       & 95.1       & 96.9                                                               & 98.0                                                                \\
bathroom-cycles-2           & 14.9      & 49.0       & 38.7       & 88.2       & 98.3                                                               & 99.6                                                                \\
bathroom2                   & 50.3      & 76.5       & 57.5       & 86.7       & 95.4                                                               & 97.6                                                                \\
bathroom                    & 23.5      & 44.2       & 73.1       & 89.2       & 93.6                                                               & 97.0                                                                \\
staircase2                  & 35.7      & 61.7       & 56.1       & 57.5       & 61.8                                                               & 62.3                                                                \\
hot-living                  & 48.1      & 80.6       & 76.0       & 86.2       & 88.4                                                               & 89.1                                                                \\
staircase                   & 17.8      & 67.2       & 72.9       & 76.8       & 84.0                                                               & 84.3                                                                \\
the-sitting-room            & 48.8      & 87.0       & 79.9       & 90.8       & 95.5                                                               & 96.5                                                                \\
dining-room                 & 20.9      & 44.9       & 91.5       & 94.7       & 95.0                                                               & 97.1                                                                \\
classroom                   & 39.3      & 62.4       & 58.7       & 66.0       & 73.4                                                               & 74.0                                                                \\
cbox                        & 33.8      & 63.0       & 94.1       & 97.3       & 97.8                                                               & 99.3                                                                \\
cycles-modern-kitchen-scene & 46.1      & 82.6       & 68.3       & 90.4       & 96.5                                                               & 97.8                                                                \\
living-room-3               & 36.1      & 62.5       & 49.3       & 80.6       & 93.9                                                               & 97.5                                                                \\
living-room-2               & 47.6      & 81.3       & 87.5       & 94.0       & 96.3                                                               & 98.7                                                                \\
my-office                   & 29.6      & 57.6       & 45.1       & 72.8       & 89.5                                                               & 94.5                                                                \\
veach-bidir                 & 28.2      & 44.1       & 41.9       & 75.6       & 99.3                                                               & 99.6                                                                \\
vgroove                     & 0.0       & 34.9       & 25.2       & 34.3       & 97.1                                                               & 90.1                                                                \\
living-room-4               & 7.5       & 28.3       & 3.4        & 31.6       & 97.2                                                               & 99.1                                                                \\
breakfast-hall              & 20.7      & 38.4       & 38.1       & 38.1       & 40.5                                                               & 40.7                                                                \\ \bottomrule
\end{tabular}
\end{table}

\clearpage
\subsection{Additional Hardware Emulation Result}

\begin{figure}[!ht]
\centering
\includegraphics[width=0.8\textwidth]{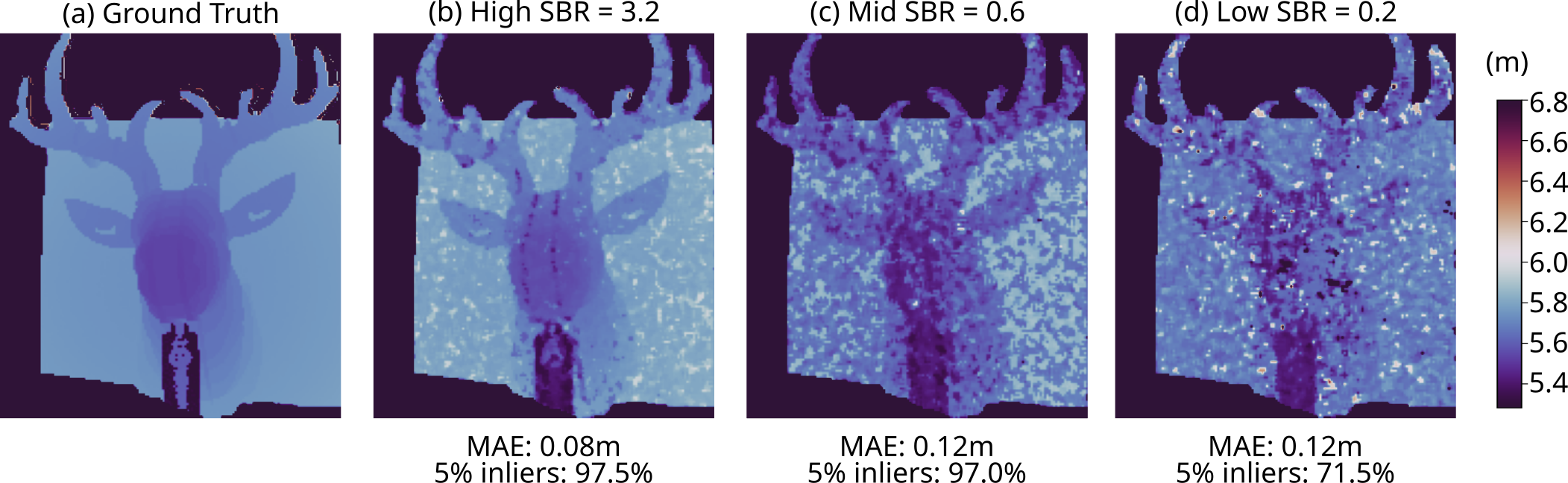}
\caption{ \rnote{ {\bf ``Elk'' distance maps under varying SBR.}
  (a) Ground truth distance map.
  (b) Distance map at high SBR shows various details such as the vertical edges on the elk's nose.
  (c) At mid SBR finer distance details are lost but the overall structure is still visible.
(d) At low SBR the distance map reconstruction suffers from severe artifacts due to background light.}
\label{fig:elk_vary_sbr}
}
\end{figure}

\clearpage
\section{Hardware Considerations: Pixel Designs and Energy Consumption}\label{app:hardware}
\subsection{Pixel and Pixel-Array Designs}

\begin{figure}[!ht]
  \centering
  \includegraphics[width=0.95\textwidth]{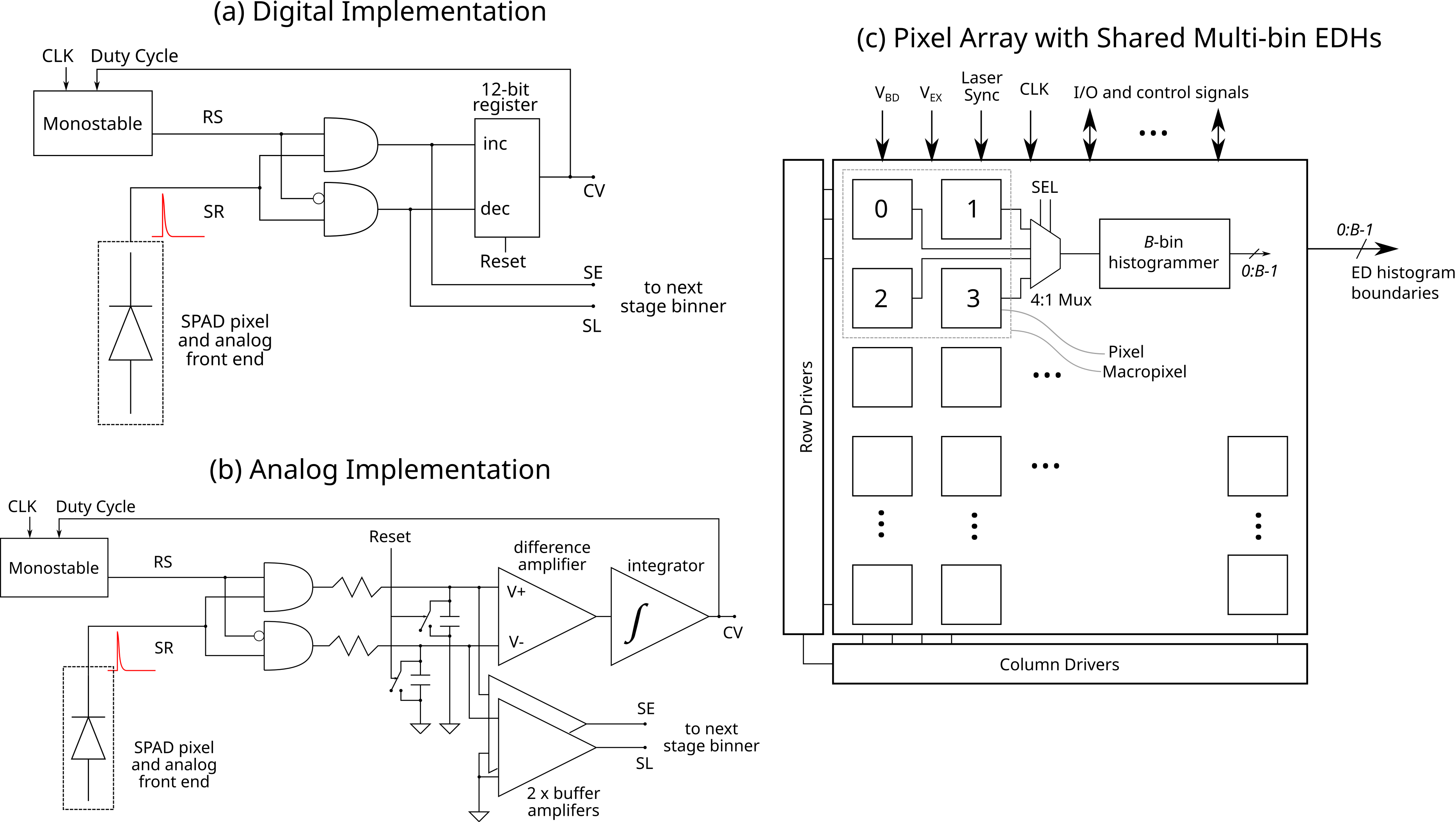}
  \caption{ {\bf Conceptual designs of binner and pixel array hardware implementations.} (a) In a digital implementation, an increment--decrement register holds the control value for an in-pixel binner circuit. 
  (b) An analog implementation the control value is updated by integrating the difference in accumulated voltage from the early and late photon streams using two capacitors.
  (c) Due to limited on-chip resources, it may not be possible to have dedicated multi-bin EDH circuit in each pixel.
  A single EDH can be shared between groups of pixels that are in spatial proximity.
\label{fig:pixel_hardware} }
\end{figure}

Suppl. Fig.~\ref{fig:pixel_hardware}(a) shows one possible digital implementation of an in-pixel binner circuit.
It uses an up--down register to represent the binner's control value.
The initial value of this register is set to the midpoint of the distance-range window.
Each early (late) photon decrements (increments) this register.
The duty cycle of the monostable circuit sets the reference signal switch point splitting the incoming photon stream into early and late streams. 
and the number of early and late photons.
The new register value is fed back to the monostable circuit to set its duty cycle.
The early and late stream readouts are available to feed to the next stage of a multi-binner EDH.
Purely analog implementations that store early and late streams as proportional voltage values and generate a difference signal as the control value may also be possible (Suppl. Fig.~\ref{fig:pixel_hardware}(b)).

In an array of SPC pixels, chip area may be limited, preventing each pixel from having multiple binner stages for a multi-bin EDH.
Groups of pixel may need to share a single EDH as shown in Suppl. Fig.~\ref{fig:pixel_hardware}(c).
For example, a spatial neighborhood (called a ``macropixel'') of $2\times 2$ pixels can be routed through a 4:1 multiplexer to a single EDH that generates a B-bin equi-depth histogram.

\subsection{Energy Consumption Estimates}
A fair comparison of energy and bandwidth requirements for different methods can be quite challenging because precise numerical estimates depend on low-level hardware implementation details.
We show some ``back-of-the-envelope'' calculations to informally motivate the idea that it may in fact be possible to radically reduce both power and bandwidth requirements by adopting EDHs instead of conventional EWHs.

One concern is that our EDH method requires longer exposure times to converge and settle at a good estimate (bin median).
With conventional EWH methods, the rule of thumb is about $100$ laser pulses.
Our empirical convergence analysis suggests that a 16-bin EDHs require an order of magnitude more laser pulses to converge (perhaps 1000, using variable step size).
Do we still save any power if we end up using more laser pulses?
The answer depends on the energy consumed by each laser pulse versus that for producing and transferring histogrammer output.
Our estimates, given in Suppl. Table~\ref{tab:energy}, suggest that an EDH can still give a $10$-fold energy savings (numbers from Morimoto et al. \cite{morimoto2020megapixel}).
We assume that a conventional EWH reads out individual photon timestamps and builds a histogram off-sensor.
(Some designs build partial histograms on sensor, but require larger in-pixel memory and still must perform multiple read-outs per run.)
Such readouts require significantly higher energy per readout than the energy per laser pulse.
Our estimate of energy per laser pulse is based on an average power of \SIrange{10}{100}{\milli\watt} for a \SI{1}{\mega\hertz} repetition rate (corresponding to an unambiguous range of \SI{150}{\meter}).
For fair comparison, we assume that the EDH is constructed using similar hardware as the conventional EWH (that is, using SPAD detectors with built-in time-to-digital conversion circuits).
These estimates will differ for alternative detector technology and hardware implementation.
We also assume that there is negligible energy consumed by any control signals that are needed to program the sensor module.

\begin{table}[!htb]
\caption{{\bf Energy consumption estimates of conventional EWH and proposed EDH.} Rough calculations suggest that there is at least a $10\times$ reduction in energy requirement with EDH versus EWH.
\label{tab:energy}}
\centering
\begin{tabular}{llll}
\hline
\multicolumn{1}{c}{\textbf{Symbol}} & \multicolumn{1}{c}{\textbf{Description}} & \multicolumn{1}{c}{\textbf{Conventional (EWH)}} & \multicolumn{1}{c}{\textbf{Proposed (EDH)}} \\ \hline
$N_\text{laser}$  & number of laser pulses         &  100      & 1000        \\
$E_\text{readout}$ & energy in each window readout  & \SIrange{1}{10}{\micro\joule}    & \SIrange{1}{10}{\micro\joule}  \\
$N_\text{readout}$     & number of readouts    & 100     & \numrange{5}{10}    \\
$E_\text{laser}$    & energy per laser pulse   & \SIrange{10}{100}{\nano\joule} & \SIrange{10}{100}{\nano\joule} \\
\hline
$N_\text{laser}E_\text{laser} + N_\text{readout} E_\text{readout}$ & total energy consumed  & \SIrange{100}{1000}{\micro\joule} & \SIrange{15}{200}{\micro\joule} \\
\hline
\end{tabular}
\end{table}

We obtain a basic estimate of bandwidth by observing that a conventional EWH must allocate enough memory to build a complete photon-count histogram consisting of around $1000$ bins, where each bin has at least an 8-bit (unsigned) integer counter.
Thus the total information (per pixel) is $8000$ bits.
The proposed EDH on the other hand only requires around $15$ 10-bit numbers to be transmitted, corresponding to $\sim 150$ bits of information per pixel, a $50 \times$ reduction in storage and bandwidth requirements.
This size differential highlights a hurdle that EWH designs must face that EDH designs do not.
The size of an EW histogram makes it infeasible to store the full histogram on sensor, but constructing the histogram off-sensor introduces energy and latency overheads.
These constraints severely limit the frame rate and resolution of today's single-photon 3D cameras.
Our EDH-based designs circumvent both of these limitations: we do not construct a full histogram and only transfer a handful of numbers off-chip.

\end{document}